\documentclass[journal]{IEEEtran}
\usepackage{epsf}
\usepackage{graphicx}
\usepackage{amsmath} \usepackage{amssymb}
\usepackage{booktabs}
\usepackage{subfigure}
\usepackage{cases}
\usepackage{color}
\newenvironment{proof}{\begin{IEEEproof}}{\end{IEEEproof}}

\newtheorem{theorem}{Theorem}[section]

\newtheorem{proposition}{Proposition}[section]
\newtheorem{lemma}{Lemma}[section]
\newtheorem{corollary}{Corollary}[section]

\long\def\symbolfootnote[#1]#2{\begingroup%
\def\thefootnote{\fnsymbol{footnote}}\footnote[#1]{#2}\endgroup}

\def\dref#1{(\ref{#1})}

\def\tan{\mbox{tan}\,} \def\sin{\mbox{sin}\,}

\def\be{\begin{equation}} \def\ee{\end{equation}}
\def\ba{\begin{array}} \def\ea{\end{array}} \def\bna{\begin{eqnarray}}
\def\ena{\end{eqnarray}}

 \def\bna{\begin{eqnarray}}
\def\ena{\end{eqnarray}} \def\dref#1{(\ref{#1})}

\begin{document}
\title{{  ``The Capacity of the Relay Channel'':\\
Solution to Cover's Problem in the Gaussian Case
}}

\author{Xiugang Wu,~\IEEEmembership{Member,~IEEE,} Leighton Pate Barnes,~\IEEEmembership{Student Member,~IEEE,}
        and~Ayfer  \"{O}zg\"{u}r,~\IEEEmembership{Member,~IEEE}

\thanks{The work was supported in part by NSF award CCF-1704624 and by the Center for Science of Information (CSoI), an NSF Science and Technology Center, under grant agreement CCF-0939370. This paper was presented in part at the 2016 Allerton Conference on Communication, Control, and Computing \cite{WuBarnesOzgur}.}

\thanks{X. Wu is with the Department of Electrical and Computer Engineering, University of Delaware, Newark, DE 19716, USA (e-mail: xwu@udel.edu). The work of X. Wu was done when he was with Stanford University.}

\thanks{L. P. Barnes and A. \"{O}zg\"{u}r are with the Department of Electrical Engineering, Stanford University, Stanford, CA 94305, USA (e-mail: lpb@stanford.edu; aozgur@stanford.edu).}

}

%

\maketitle

\begin{abstract}
Consider a memoryless relay channel, where  the relay is connected to the destination with an isolated bit pipe of capacity $C_0$. Let $C(C_0)$ denote the capacity of this channel as a function of $C_0$. What is the critical value of $C_0$ such that $C(C_0)$ first equals $C(\infty)$? This is a long-standing open problem posed by Cover and named ``The Capacity of the Relay Channel,'' in \emph{Open Problems in Communication and Computation}, Springer-Verlag, 1987. In this paper, we answer this question in the Gaussian case and show that $C(C_0)$ can not equal to $C(\infty)$ unless $C_0=\infty$, regardless of the SNR of the Gaussian channels. This result follows as a corollary to a new upper bound we develop on the capacity of this channel. Instead of ``single-letterizing'' expressions involving information measures in a high-dimensional space as is typically done in converse results in information theory, our proof directly quantifies the tension between the pertinent $n$-letter forms. This is done by translating the information tension problem to a problem in high-dimensional geometry. As an intermediate result, we develop an extension of the classical isoperimetric inequality on a high-dimensional sphere, which can be of interest in its own right. 
\end{abstract}


\section{Problem Setup and Main Result}\label{S:Introduction}

{ In 1987, Thomas M. Cover formulated a seemingly simple question in \emph{Open Problems in Communication and Computation}, Springer-Verlag \cite{cg87}, which he called ``The Capacity of the Relay Channel''. This problem, not much longer than a single page in \cite{cg87}, remains open to date.  His problem statement, taken verbatim from \cite{cg87} with only a few minor notation changes, is as follows:}

\begin{center}\textit{The Capacity of the Relay Channel}
\end{center}
\textit{Consider the following seemingly simple discrete memoryless relay channel:}
\begin{figure}[h!]
\centering
\includegraphics[width=0.3\textwidth]{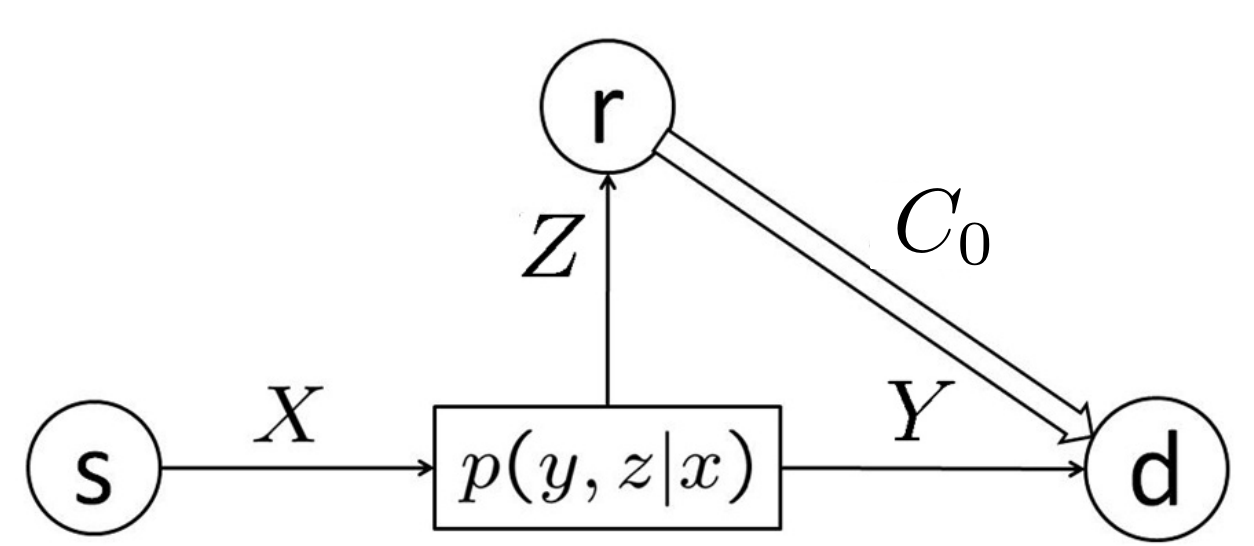}
\end{figure}
\textit{
Here $Z$ and $Y$ are conditionally independent and conditionally identically distributed given $X$, that is, $p(z,y|x)=p(z|x) p(y|x)$. Also, the channel from $Z$ to $Y$ does not interfere with $Y$. A $(2^{nR},n)$ code for this channel is a map $X^n: [1:2^{nR}]\to \mathcal X^n$, a relay function $f_n: \mathcal Z^n \to [1:2^{nC_0}]$ and a decoding function $g_n: \mathcal Y^n \times  [1:2^{nC_0}] \to [1:2^{nR}]$. The probability of error is given by
$$P_e^{(n)}=\mbox{Pr}(g_n(Y^n,f_n(Z^n)) \neq M ),$$
where the message $M$ is  uniformly distributed over $[1:2^{nR}]$ and
$$
p(m,y^n,z^n)=2^{-nR}\prod_{i=1}^n p(y_{i}|x_{i}(m))\prod_{i=1}^n p(z_{i}|x_{i}(m)).
$$
Let $C(C_0)$ be the supremum of achievable rates $R$ for a given $C_0$, that is, the supremum of the rates $R$ for which $P_e^{(n)}$ can be made to tend to zero.
We note the following facts:
\begin{itemize}
\item[1.] $C(0)=\sup_{p(x)} I(X;Y).$
\item[2.] $C(\infty)=\sup_{p(x)} I(X;Y,Z).$
\item[3.] $C(C_0)$ is a nondecreasing function of $C_0$.
\end{itemize}
What is the critical value of $C_0$ such that $C(C_0)$ first equals $C(\infty)$?
}
\subsection{Main Result}
{ As is customary in network information theory, Cover formulates the problem for discrete memoryless channels. However, the same question clearly applies to channels with continuous input and output alphabets, and in particular when the channels from the source to the relay and the destination are Gaussian, which is the canonical model for wireless relay channels. More formally, assume}
\begin{numcases}{}
Z=X+W_1\nonumber \\
Y=X+W_2 \nonumber
\end{numcases}
with the transmitted signal being constrained to average power $P$, i.e.,
\begin{equation}\label{avpowerconst}
  \|x^n(m)\|^2\leq nP, \  \forall m \in [1:2^{nR}],
\end{equation}
and $W_1, W_2 \sim \mathcal N(0,N)$ representing Gaussian noises that are independent of each other and $X$. See Fig. \ref{F:GaussianRelay}.
\begin{figure}[htb!]
\centering
\includegraphics[width=0.35\textwidth]{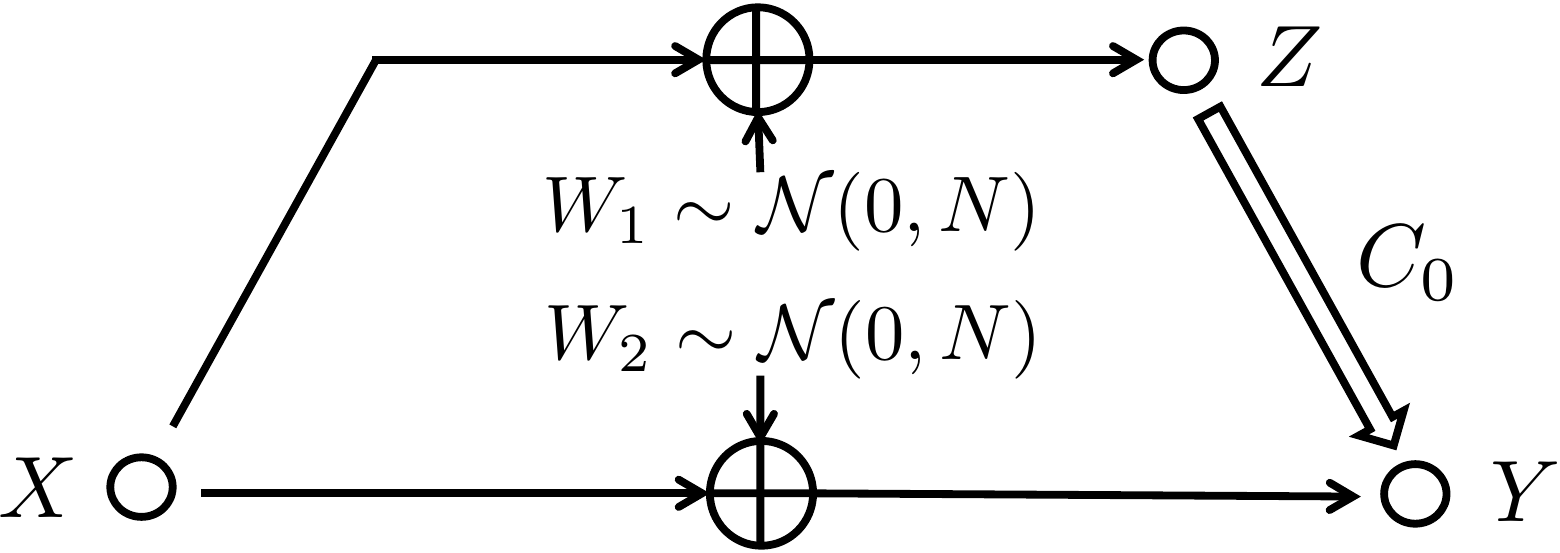}
\caption{Symmetric Gaussian relay channel.}
\label{F:GaussianRelay}
\end{figure}

For this Gaussian relay channel, it is easy to observe that\footnote{All logarithms throughout the paper are to base two.} $$C(\infty)=\frac{1}{2}\log \left( 1+\frac{2P}{N}\right).$$
{Let $C_0^*$ denote the threshold in Cover's problem, i.e.
\begin{align}
C_0^*:=\mbox{inf}\{C_0: C(C_0)=C(\infty) \}.
\end{align}
For the Gaussian model, there is no known scheme that allows to achieve $C(\infty)$ at a finite $C_0$ regardless of the parameters of the channels, i.e. the signal to noise power ratio (SNR) $P/N$. Therefore, from an achievability perspective we only have the trivial bound
$$
C_0^*\leq\infty.
$$
On the converse side, any upper bound on the capacity of this channel can be used to establish a lower bound on $C_0^*$. The only upper bound on the capacity of this channel (prior to our  work in \cite{Allerton2015}--\cite{WuOzgur_TIT_Gaussian} preceding the current paper) was the celebrated cut-set bound developed by Cover and El Gamal in 1979 \cite{covelg79}. It yields the following lower bound on $C_0^*$:
$$
C_0^*\geq \frac{1}{2}\log \left( 1+\frac{2P}{N}\right)-\frac{1}{2}\log \left( 1+\frac{P}{N}\right).
$$
Note that the cut-set bound does not preclude achieving $C(\infty)$ at finite $C_0$. Moreover, it is interesting to note that as $P/N$ decreases to zero, this lower bound decreases to zero. This implies a sharp dichotomy between the current achievability and converse results for this problem, which becomes even more apparent in the limit when SNR goes to zero: the cut-set bound does not preclude achieving $C(\infty)$ at diminishing $C_0$ if $C(\infty)$ itself is diminishing, while from an achievability perspective we need $C_0=\infty$ regardless of the SNRs of the channels (apart from the trivial case when $P/N$ is exactly equal to $0$). The main result of our paper is to show that $C_0^*=\infty$ regardless of the parameters of the problem, answering Cover's long-standing question for the canonical Gaussian model.
\smallbreak
\begin{theorem}\label{thm:maintheorem}
For the symmetric Gaussian  relay channel depicted in Fig. \ref{F:GaussianRelay},  $C^*_0=\infty$.
\end{theorem}
\medbreak
This theorem follows immediately from the following theorem which establishes a new upper bound on the capacity of this channel for any $C_0$.
\smallbreak
\begin{theorem}\label{thm:maintheorem2}
For the symmetric Gaussian  relay channel depicted in Fig. \ref{F:GaussianRelay}, the capacity $C(C_0)$ satisfies
\begin{align*} 
 C(C_0)  \leq \ &\frac{1}{2}\log \left(1+\frac{P}{N}\right)\\
 &+\sup_{\theta\in\left[\arcsin (2^{-C_0}),\frac{\pi}{2}\right]}\min \Bigg\{ 
 \begin{split}
 & C_0+\log \sin \theta, \\
& \min_{\omega\in \left(\frac{\pi}{2}-\theta, \frac{\pi}{2}\right]}  h_{\theta}(\omega)
\end{split}
\Bigg\}
\end{align*}

where
\begin{align*}&~h_{\theta}(\omega) :=    \frac{1}{2}\log \left(\frac{4\text{sin} ^2\frac{\omega}{2}(P+N-N\text{sin}^2 \frac{\omega}{2})\sin^2\theta}{(P+N)(\text{sin} ^2 \theta - \cos^2 \omega)} \right).
 \end{align*}
\end{theorem}

In Fig.~\ref{fig:plots} we plot this upper bound (label: New bound) under three different SNR values of the Gaussian channels, together with the cut-set bound \cite{covelg79} and an upper bound on the capacity of this channel we have previously derived in \cite{WuOzgur_TIT_Gaussian} (label: Old bound). For reference, we also provide the rate achieved by a compress-and-forward relay strategy (label: C-F), which employs Gaussian input distribution at the source combined with Gaussian quantization and Wyner-Ziv binning at the relay.\footnote{In the low SNR regime, we can achieve higher rates using bursty compress-and-forward \cite{ElGamalKim}, as demonstrated in the left-most plot of Fig.~\ref{fig:plots}. Note that since we still impose the Gaussian restriction on the input and quantization distributions for bursty compress-forward,  the resultant rates are not concave in $C_0$ and can be further improved by time sharing.} The flat levels at which the cut-set bound and our old bound saturate in these plots precisely correspond to $C(\infty)$. Note that while these earlier bounds reach $C(\infty)$ at finite $C_0$ values, hence leading to finite lower bounds on $C_0^*$, our new bound remains bounded away from $C(\infty)$ in all the three plots. Indeed, it can be formally shown that the new bound remains bounded away from $C(\infty)$ (the flat level in the plots) at any finite $C_0$ value. We prove this formally in the proof of Theorem~\ref{thm:maintheorem}.

\begin{figure*}[hbt]
\centering
\subfigure{\includegraphics[width=2.3in]{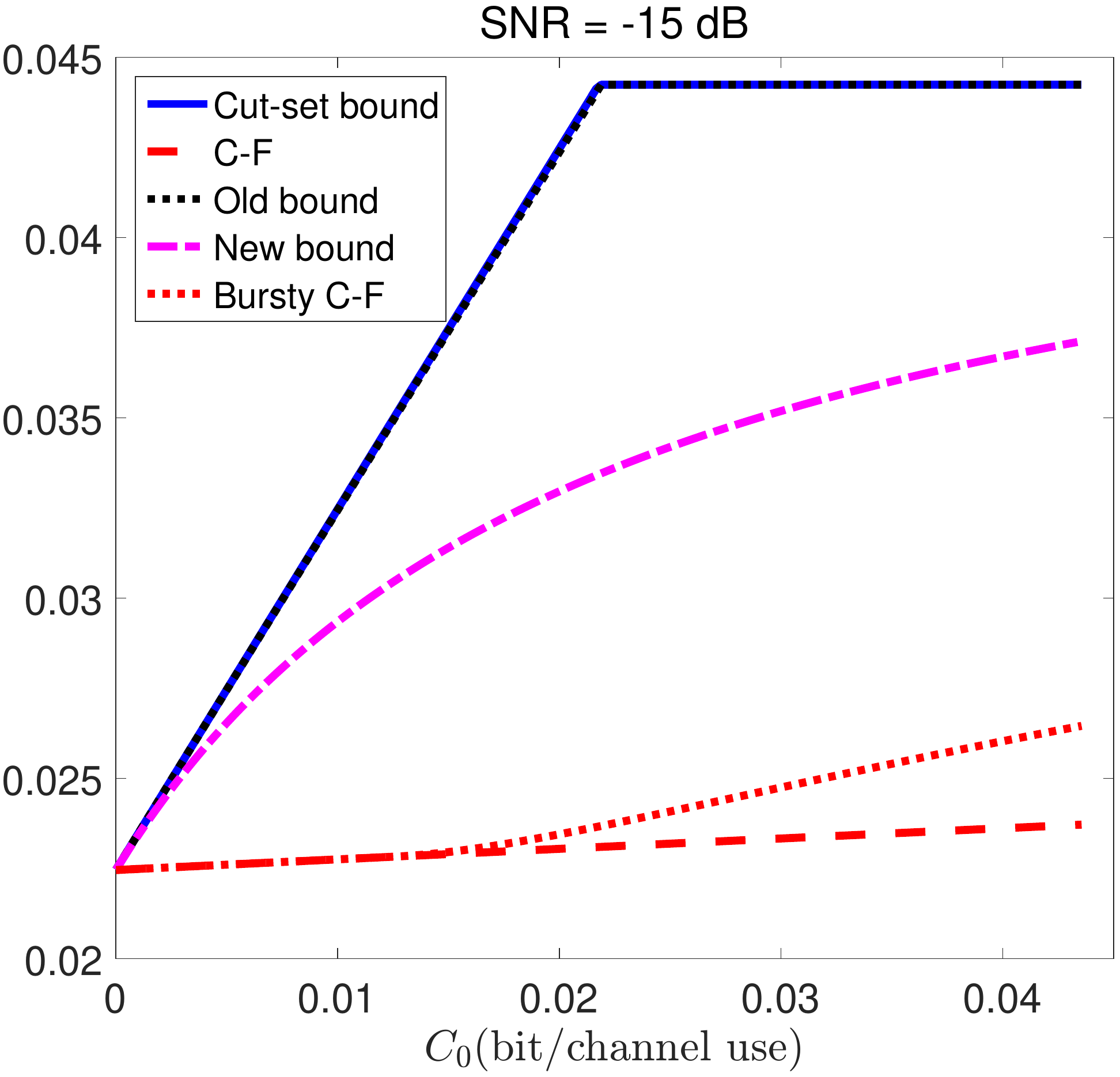}}
\subfigure{\includegraphics[width=2.31in]{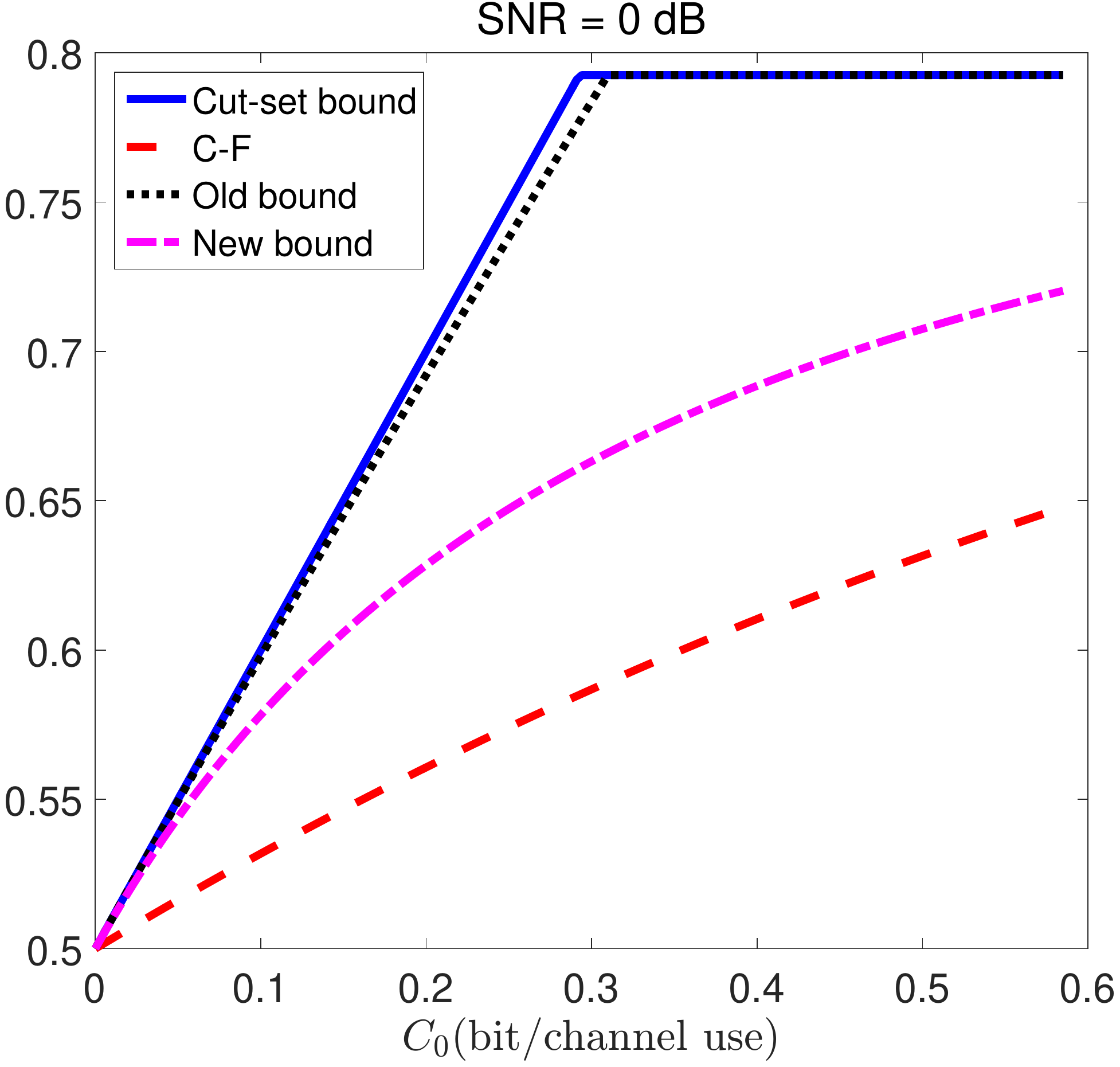}}
\subfigure{\includegraphics[width=2.26in]{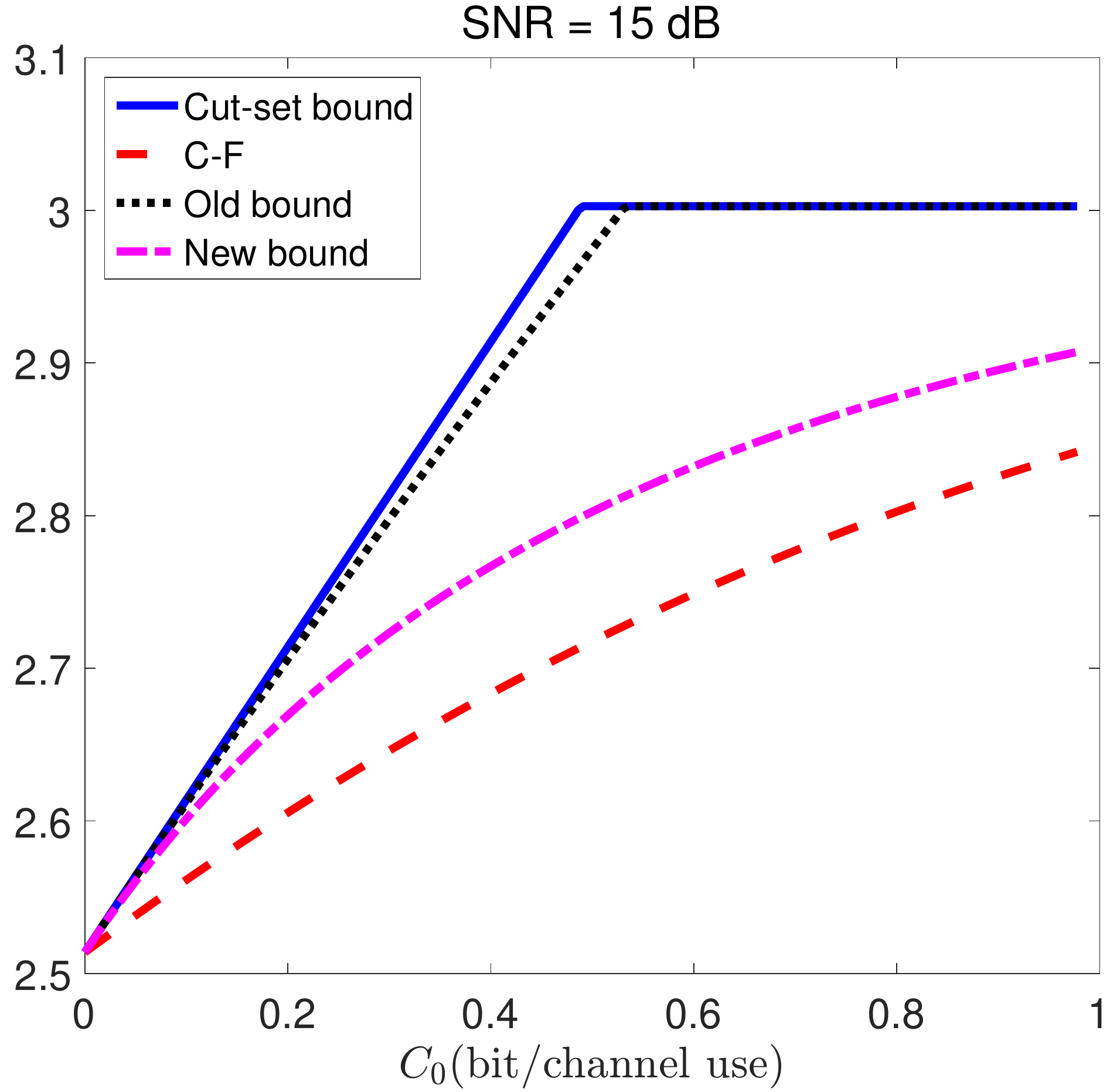}}
\caption{Upper bounds and achievable rates for the Gaussian relay channel.}
\label{fig:plots}
\end{figure*}
While in this paper we restrict our attention to the symmetric case, an assumption imposed by Cover in his  original formulation of the problem given above, our methods and results also extend to the asymmetric case. In \cite{ISIT2017}, we show that when the relay's and the destination's observations are corrupted by independent Gaussian noises of different variances, it is still true that $C_0^*=\infty$ regardless of the channel parameters. The extension to this asymmetric case heavily builds on the methods and results we develop in this paper for the symmetric case. Interestingly, the symmetric case, which Cover seems to somewhat arbitrarily assume in his problem formulation, turns out to be the canonical case for our proof technique. We also provide a solution to Cover's problem for binary symmetric channels in \cite{Allerton2017} using a similar approach.

\subsection{Technical Approach}
There are two basic aspects in an information-theoretic characterization of an operational
problem: the so-called achievability result and converse result. An achievability result establishes what
is possible in a given setting, while the converse result distinguishes what is impossible. The ideal situation
is when these two results match, in which case an information limit is born. The most famous example goes
back to Shannon and the inception of the field: Reliable communication is possible over a noisy channel if, and only if, the rate of transmission does not exceed the capacity of the channel \cite{Shannon1948}.

Over the last two decades, there has been significant leap forward in developing achievable schemes for multi-user problems, ranging from schemes based on interference alignment and distributed MIMO, to lattice-based techniques, to strategies inspired by network coding and linear deterministic models. This stands in fairly stark contrast to the set of
converse arguments in the information theorist's toolkit. Almost all converse arguments
rely on a few fundamental tools that go back to the early years of the field: information measure calculus (e.g., chain rules, non-negativity of divergence), Fano's inequality, and the entropy power inequality.
The typical converse program follows  from a clever application of these tools to ``single-letterize'' an expression involving information measures in a high-dimensional space (so called $n$-letter forms), with the possible introduction of auxiliary random variables as needed.

In this paper, we take a different approach. Instead of focusing on single-letterizing pertinent $n$-letter forms, we aim to directly quantify the tension between them. To do this, we lift the problem to an even higher dimensional space and study the geometry of the typical sequences  generated independently and identically (i.i.d.) from these $n$-dimensional distributions. We establish non-trivial geometric properties satisfied by these typical sequences, which are then translated  to inequalities satisfied by the original $n$-dimensional information measures. This notion of ``typicality'', connecting information measures associated with a distribution to probabilities of  long i.i.d. sequences generated from this distribution, is a standard tool in establishing achievability results in information theory but to the best of our knowledge has been rarely used in proving converse results in network information theory, with only a few examples such as the work of Zhang \cite{Zhang} from 1988 and our recent works \cite{IZS2016}--\cite{ISIT2016}.

To study the geometry of the typical sequences, we use classical tools from high-dimensional geometry, such as the isoperimetric inequality \cite{levy}, measure concentration \cite{mono}, and rearrangement and symmetrization theory \cite{shortcourse,baernstein}. We also prove a new geometric result which can be regarded as an extension of the classical isoperimetric inequality on a high-dimensional sphere and can be of interest in its own right. Note that the classical isoperimetric inequality on the sphere states that among all sets on the sphere with a given measure (area), the spherical cap has the smallest boundary  or more generally the smallest  neighborhood \cite{isoperimetric}. As an intermediate result in this paper, we show that the spherical cap not only minimizes the measure of its neighborhood, but roughly speaking, also minimizes the measure of its intersection with the neighborhood of a randomly chosen point on the sphere.

The incorporation of geometric insight in information theory is not new. Formulating the problem
of determining the communication capacity of channels as a problem in high-dimensional
geometry is indeed one of Shannon's most important insights that has led to the conception of
the field. In his classical paper ``Communication in the presence of noise'', 1949 \cite{shannon49},
Shannon develops a geometric representation of any point-to-point communication system, and
then uses this geometric representation to derive the
capacity formula for the AWGN channel. His converse proof is based on a sphere-packing
argument, which relies on the notion of sphere hardening (i.e. measure concentration)  in high-dimensional space. Our approach resembles Shannon's approach in \cite{shannon49} in that the main argument in our proof is also a packing argument; however, instead of packing smaller spheres in a larger sphere, we  pack (quantization) regions of some minimal measure (and unknown shape) inside a spherical cap. The key ingredient in our packing argument is the extended isoperimetric inequality we develop, which guarantees that each of these quantization regions has some minimal intersection with the spherical cap. Also, note that we do not directly study the geometry of the codewords as in \cite{shannon49}, but rather use geometry in an indirect way to solve an $n$-letter information tension problem. 
}
\subsection{Organization of The Paper} The remainder of the paper is organized as follows. In Section~\ref{sec:geometry}, we review some basic definitions and results for high-dimensional spheres, and state our main geometric result in Theorem~\ref{thm:strongisoperimetryshell}, which can be regarded as an extension of the classical isoperimetric inequality on the sphere. In Section~\ref{sec:inequality}, we introduce some typicality lemmas and combine them with Theorem~\ref{thm:strongisoperimetryshell} to prove a key information inequality stated in Theorem~\ref{L:upperboundlemma}. The proofs of our main theorems,  Theorem \ref{thm:maintheorem} and \ref{thm:maintheorem2}, are almost immediate given Theorem~\ref{L:upperboundlemma} and are provided in Section~\ref{S:proof}. 

Appendices \ref{A:proofisop} and \ref{A:Proof_Typicality} are then devoted to the proof of Theorem~\ref{thm:strongisoperimetryshell} and the proofs of the typicality lemmas introduced in Section~\ref{sec:inequality}, respectively.  The proofs of these typicality lemmas require us to derive formulas and exponential characterizations for the area/volume of various high dimensional sets including balls, spherical caps, shell caps, and intersections of such sets. We derive these characterizations in Appendix~\ref{A:Miscellaneous}. 

\section{Geometry of High-Dimensional Spheres}\label{sec:geometry}

In this section, we summarize some basic definitions and results for high-dimensional spheres and present our main geometric result which can be regarded as an extension of the classical isoperimetric inequality on high-dimensional spheres. This result is the key to proving the information inequality we present in the next section, which in turn is the key to proving Theorems~\ref{thm:maintheorem} and \ref{thm:maintheorem2}. 

\subsection{Basic Results on High-Dimensional Spheres}\label{sec:geometricprelim}

We now summarize some basic results on high-dimensional spheres that will be referred to later in the paper.

\begin{itemize}
\item[(i)]{\bf Isoperimetric Inequality:}  Let $\mathbb{S}^{m-1}\subseteq \mathbb{R}^m$ denote the $(m-1)$-sphere of radius $R$, i.e., $$\mathbb{S}^{m-1}=\left\{\mathbf z\in\mathbb{R}^m:\|\mathbf z\|=R\right\},$$ equipped with the rotation invariant (Haar) measure $\mu=\mu_{m-1}$ that is normalized such that $$\mu(\mathbb{S}^{m-1}) = \frac{2\pi^{\frac{m}{2}}}{\Gamma(\frac{m}{2})}R^{m-1},$$
i.e. the usual surface area. Let $\mathbb{P}(A)$ denote the probability of a set or event $A$ with respect to the corresponding Haar probability measure, i.e. the normalized Haar measure such that $\mathbb{P}(\mathbb{S}^{m-1})=1$.
A spherical cap is defined as a ball on $\mathbb{S}^{m-1}$ in the geodesic metric (or simply the angle) $\angle( \mathbf z,\mathbf y)=\arccos(\langle\mathbf z/R,\mathbf y/R\rangle)$, i.e.,
$$
\text{Cap}(\mathbf z_0,\theta)=\left\{\mathbf z\in \mathbb{S}^{m-1}: \angle(\mathbf z_0,\mathbf z)\leq \theta \right\}.
$$
See Fig.~\ref{F:onecap}. We will often say that an arbitrary set $A\subseteq \mathbb{S}^{m-1}$ has an effective angle $\theta$ if $\mu(A)=\mu(C)$, where $C=\text{Cap}(\mathbf z_0,\theta)$ for some arbitrary $\mathbf z_0\in\mathbb{S}^{m-1}$.

\begin{figure}[t!]
\centering
\includegraphics[width=0.3\textwidth]{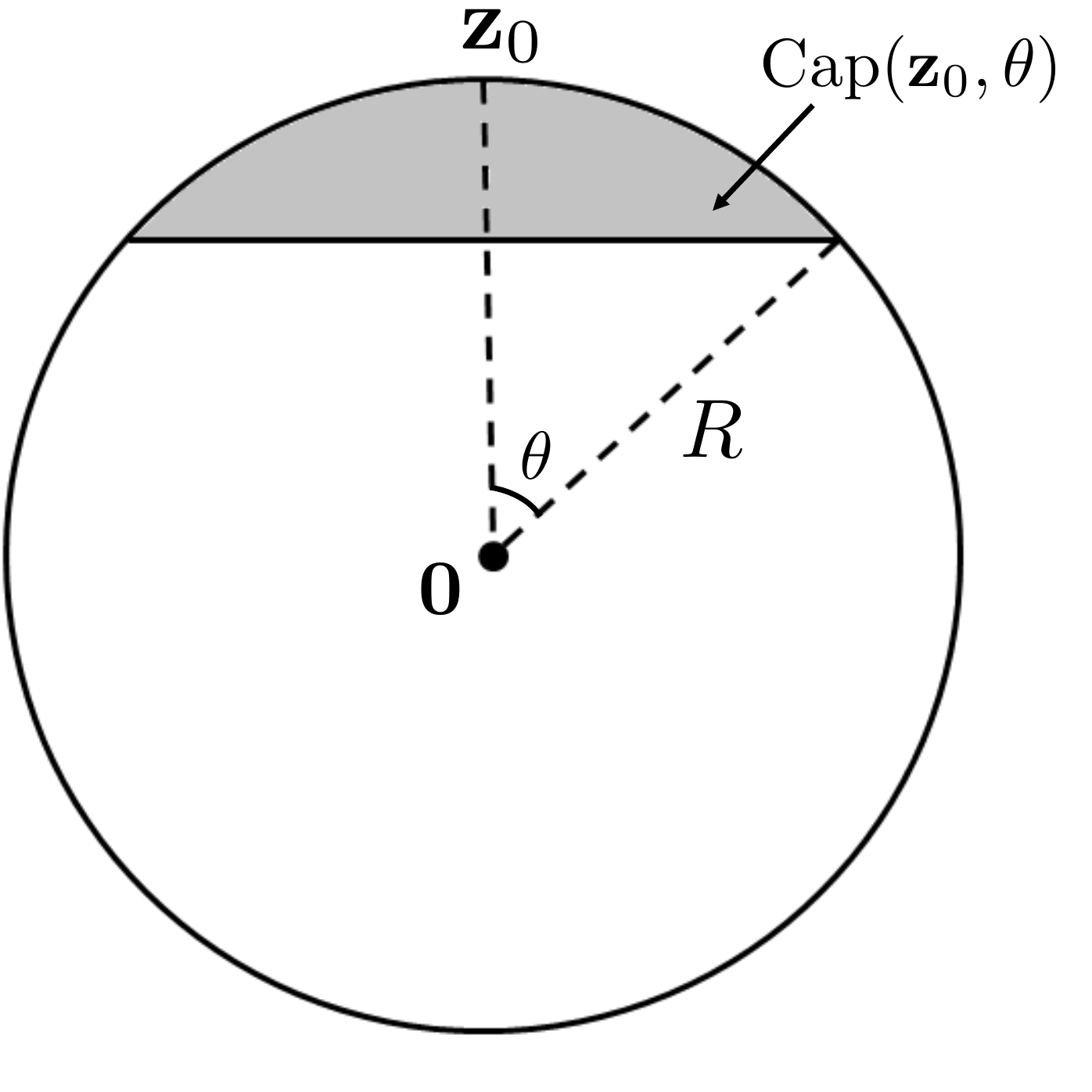}
\caption{A spherical cap with angle $\theta$.}
\label{F:onecap}
\end{figure}
\medbreak
The following proposition is the so-called isoperimetric inequality, which was first proved by Levy in 1951 \cite{levy}.  (See also \cite{isoperimetric}.) It states the intuitive fact that among all sets on the sphere with a given  measure, the spherical cap has the smallest boundary, or more generally the smallest neighborhood. This is formalized as follows:
\medbreak
\begin{proposition}\label{prop:isoperimetry}
For any arbitrary set $A\subseteq \mathbb{S}^{m-1}$ such that $\mu(A)=\mu(C)$, where $C=\text{Cap}(\mathbf z_0,\theta)\subseteq \mathbb{S}^{m-1}$ is a spherical cap, it holds that $$\mu(A_t)\geq \mu(C_t), \ \forall t\geq 0,$$ where $A_t$ is the $t$-neighborhood of $A$, defined as $$A_t=\left\{\mathbf z\in \mathbb{S}^{m-1}: \min_{\mathbf z'\in A}\angle(\mathbf z, \mathbf z')\leq t\right\},$$ and similarly $$C_t=\left\{\mathbf z\in \mathbb{S}^{m-1}: \min_{\mathbf z'\in C}\angle(\mathbf z, \mathbf z')\leq t\right\}=\text{Cap}(\mathbf z_0,\theta+t).$$
\end{proposition}

\item[(ii)]{\bf Measure Concentration:} Measure concentration on the sphere refers to the fact that most of the measure of a high-dimensional sphere is concentrated around any equator. The following elementary result capturing this phenomenon will be used later in the paper when we prove the extended isoperimetric inequality.
\medbreak
\begin{proposition} \label{P:measurecon}
Given any $\epsilon,\delta>0$, there exists some $M(\epsilon, \delta)$ such that for any $m\geq M(\epsilon, \delta)$ and any $\  \mathbf z \in \mathbb{S}^{m-1}$,
\begin{align}
\mathbb P \left( \angle(\mathbf z,\mathbf Y) \in [\pi/2-\epsilon,\pi/2+\epsilon] \right)\geq 1-\delta , \end{align}
where $\mathbf Y\in\mathbb{S}^{m-1}$ is distributed according to the Haar probability measure.
\end{proposition}
\begin{proof}
Let $\mathbf e_1=(R,0,\ldots, 0)$. Note for any $\mathbf z \in \mathbb{S}^{m-1}$, the distribution of $\angle(\mathbf z,\mathbf Y)$ is the same as the distribution of $\angle(\mathbf e_1,\mathbf Y)$, since $\mathbf z$ can be written in the form $\mathbf z=U\mathbf e_1$, where $U$ is an orthogonal matrix, and the distribution of $\mathbf Y$ is rotation-invariant. Therefore, without loss of generality, we can assume $\mathbf z=\mathbf e_1$. Since $\langle \mathbf e_1/R, \mathbf Y/R\rangle=Y_1/R$,  we have $ E[\langle \mathbf e_1/R, \mathbf Y/R\rangle]= E[Y_1]/R=0$; we also have $  E[\langle \mathbf e_1/R, \mathbf Y/R\rangle^2]=  E[Y_1^2]/R^2=1/m$ because $E[Y_1^2]=\cdots=E[Y^2_m]$ and $E[Y_1^2]+\cdots+E[Y^2_m]=R^2$. Therefore by Chebyshev's inequality, for any $\mu>0$,
\begin{equation*}
\mathbb P(|\langle \mathbf e_1/R, \mathbf Y/R\rangle|\geq \mu)\leq \frac{1}{m\mu^2}.
\end{equation*}
Recalling that $\angle( \mathbf e_1,\mathbf Y)=\arccos(\langle\mathbf e_1/R,\mathbf Y/R\rangle)$ and noting that the R.H.S. of the above inequality can be made arbitrarily small by choosing $m$ to be sufficiently large, we have proved the proposition.
\end{proof}
\medbreak

\item[(iii)]{\bf Blowing-Up Lemma:} The above measure concentration result combined with the isoperimetric inequality immediately yields the following result:
\medbreak
\begin{proposition}\label{P:blowup}
Let $A\subseteq \mathbb{S}^{m-1}$ be an arbitrary set and $C=\text{Cap}(\mathbf z_0,\theta)\subseteq \mathbb{S}^{m-1}$ be a spherical cap such that $\mu(A)=\mu(C)$, i.e. $A$ has an effective angle of $\theta$. Then for any $\epsilon>0$ and $m$ sufficiently large,
\begin{equation}\label{eq:blowup0}
\mathbb{P}(A_{\frac{\pi}{2}-\theta+\epsilon})\geq 1-\epsilon.
\end{equation}
\end{proposition}

\begin{proof}
If $A=\text{Cap}(\mathbf z_0,\theta)$, $\mathbb{P}(A_{\frac{\pi}{2}-\theta+\epsilon})\geq 1-\epsilon$ due to Proposition \ref{P:measurecon}. If $A$ is not a spherical cap, then $\mathbb{P}(A_{\frac{\pi}{2}-\theta+\epsilon})\geq P(C_{\frac{\pi}{2}-\theta+\epsilon})$  where $C=\text{Cap}(\mathbf z_0,\theta)$, due to the isoperimetric inequality in Proposition~\ref{prop:isoperimetry}.
\end{proof}
\medbreak

If we take $A$ to be a half sphere, this result says that most of the measure of the sphere is concentrated around the boundary of this half-sphere, i.e. an equator, which is the result in Proposition \ref{P:measurecon}. However, due to the isoperimetric inequality,  Proposition~\ref{P:blowup} allows us to make the stronger statement that the measure is concentrated around the boundary of any set with probability $1/2$. While the elementary results we establish above suggest that this concentration takes place at a polynomial speed in the dimension $m$, it can be shown that the measure concentrates around the boundary of any set with probability $1/2$ exponentially fast in the dimension $m$; see \cite{matou}.
%

\end{itemize}

\subsection{Extended Isoperimetry on the Sphere and the Shell}

An almost equivalent way to state the blowing-up lemma in Proposition~\ref{P:blowup} is the following: Let $A\subseteq \mathbb{S}^{m-1}$ be an arbitrary set with effective angle $\theta>0$. Then for any $\epsilon>0$ and sufficiently large $m$,
\begin{equation}\label{eq:isoperimetry}
\mathbb{P}\left(\mu\left(A\cap \text{Cap}\left(\mathbf Y,\frac{\pi}{2}-\theta+\epsilon\right)\right)>0\right)> 1-\epsilon,
\end{equation}
{where $\mathbf Y$ is distributed according to the normalized Haar measure on $\mathbb{S}^{m-1}$.} In words, if we take a $\mathbf y$ uniformly at random on the sphere and draw a spherical cap of angle slightly larger than $\frac{\pi}{2}-\theta$ around it, this cap will intersect the set $A$ with high probability. This statement is almost equivalent to \eqref{eq:blowup0} since the $\mathbf y$'s for which the intersection has non-zero measure lie in the $\frac{\pi}{2}-\theta+\epsilon$-neighborhood of $A$. Note that similarly to Proposition~\ref{P:blowup}, this statement would trivially follow from measure concentration on the sphere (Proposition~\ref{P:measurecon}) if $A$ were known to be a spherical cap, and it holds for any $A$ due to the isoperimetric inequality in Proposition~\ref{prop:isoperimetry}. By building on the Riesz rearrangement inequality \cite{baernstein}, we prove the following extended result:

\medbreak
\begin{theorem}\label{thm:strongisoperimetrysphere}
 Let $A\subseteq \mathbb{S}^{m-1}$ be any arbitrary subset of $\mathbb{S}^{m-1}$ with effective angle $\theta>0$, and let $V=\mu(\text{Cap}(\mathbf z_0, \theta) \cap \text{Cap}(\mathbf y_0, \omega))$ where  $\mathbf z_0,\mathbf y_0 \in \mathbb{S}^{m-1}$ with $\angle(\mathbf z_0,\mathbf y_0)=\pi/2$
and $\theta+\omega>\pi/2$. (See Fig.~\ref{F:twocap}.) Then for any $\epsilon>0$, there exists an $M(\epsilon)$ such that for $m > M(\epsilon)$,
\begin{equation*}
\mathbb{P}\left(\mu(A\cap \text{Cap}(\mathbf Y,\omega+\epsilon))> (1-\epsilon)V\right)\geq 1-\epsilon,
\end{equation*}
where $\mathbf Y$ is a random vector on $\mathbb{S}^{m-1}$ distributed according to the normalized Haar measure.
\end{theorem}
\medbreak

If $A$ itself is a cap, then the statement in Theorem~\ref{thm:strongisoperimetrysphere} is straightforward and follows from the fact that $\mathbf Y$ with high probability will be concentrated around the equator at angle $\pi/2$ from the pole of $A$ (Proposition~\ref{P:measurecon}). Therefore, as $m$ gets large for almost all $\mathbf Y$, the intersection of the two spherical caps will be given by $V$. See Fig.~\ref{F:twocap}. The statement, however, is stronger than this and holds for any arbitrary set $A$, analogous to the isoperimetric inequality in \eqref{eq:isoperimetry}. It states that no matter what the set $A$ is, if we take a random point on the sphere and draw a cap of angle slightly larger than $\omega$ centered at this point, for any $\omega>\pi/2-\theta$, then with high probability the intersection of the cap with the set $A$ would be at least as large as the intersection we would get if $A$ were a spherical cap. In this sense, Theorem~\ref{thm:strongisoperimetrysphere} can be regarded as an extension of the isoperimetric inequality in Proposition~\ref{prop:isoperimetry}, even though the latter can be stated purely geometrically and implies the weaker probabilistic statement in \eqref{eq:isoperimetry}, while our result is inherently probabilistic.

Theorem \ref{thm:strongisoperimetrysphere} is in fact a special case of a more general theorem that is true for subsets on a spherical shell. Let 
$$\mathbb{L}^{m} = \{\mathbf y\in \mathbb{R}^m : \; R_L \leq \|\mathbf y\| \leq R_U\}$$
be this shell, where $0\leq R_L\leq R_U$. A cap on this shell with pole  $\mathbf z_0$ and angle $\theta$ can be defined as a ball in terms of the angle:
$$ \angle (\mathbf y,\mathbf z)=\arccos\left(\frac{\mathbf y \cdot \mathbf z }{\|\mathbf y\|   \|\mathbf z\|}\right)$$ on the shell, i.e.,
\begin{align*}
\text{ShellCap}(\mathbf z_0,\theta)=\left\{ \mathbf z\in \mathbb{L}^{m}:  \angle(\mathbf z_0,\mathbf z)\leq \theta\right\}.
\end{align*}
Let $|A|$ denote the standard $m$-dimensional Euclidean measure of a subset $A\subseteq\mathbb{L}^m$. We will say that an arbitrary set $A\subseteq \mathbb{L}^{m}$ has effective angle $\theta>0$ if its measure is equal to that of a shell cap of angle $\theta$, i.e. $|A|=|\text{ShellCap}(\mathbf z_0, \theta)|$ for some $\mathbf z_0\in \mathbb{L}^{m}$. We will also say that a probability measure $\mathbb{P}$ for subsets of $\mathbb{L}^m$ is rotationally invariant if $\mathbb{P}(A) = \mathbb{P}(UA)$ for any orthogonal matrix $U$, where $UA$ denotes the image of the set $A$ under the linear transformation $U$. The following more general theorem holds in the shell setting.

\medbreak
\begin{theorem}\label{thm:strongisoperimetryshell}
Let $A\subseteq \mathbb{L}^{m}$ be any arbitrary subset of $\mathbb{L}^{m}$ with effective angle $\theta>0$, and let $V=|\text{ShellCap}(\mathbf z_0, \theta) \cap \text{ShellCap}(\mathbf y_0, \omega)|$ where  $\mathbf z_0,\mathbf y_0 \in \mathbb{L}^{m}$ with $\angle(\mathbf z_0,\mathbf y_0)=\pi/2$
and $\theta+\omega>\pi/2$. Then for any $\epsilon>0$, there exists an $M(\epsilon)$ such that for $m > M(\epsilon)$, 
\begin{equation*}
\mathbb{P}\left(|A\cap \text{ShellCap}(\mathbf Y,\omega+\epsilon)|> (1-\epsilon)V\right)\geq 1-\epsilon,
\end{equation*}
where $\mathbf Y$ is a random vector drawn from any rotationally invariant probability measure on $\mathbb{L}^{m}$.
\end{theorem}
\medbreak
	
We prove Theorems \ref{thm:strongisoperimetrysphere} and \ref{thm:strongisoperimetryshell} in Appendix \ref{A:proofisop}. {Note that $M(\epsilon)$ in these two results depends only on $\epsilon$---in particular it does not depend on the radius parameters for $\mathbb{L}^{m}$ and $\mathbb{S}^{m-1}$, respectively, which means that these two results also apply if the radius parameters depend on the dimension $m$. In the following section, we will be mainly interested in the case when the radius parameters scale in the square-root of the dimension. } 

\begin{figure}[t!]
\centering
\includegraphics[width=0.40\textwidth]{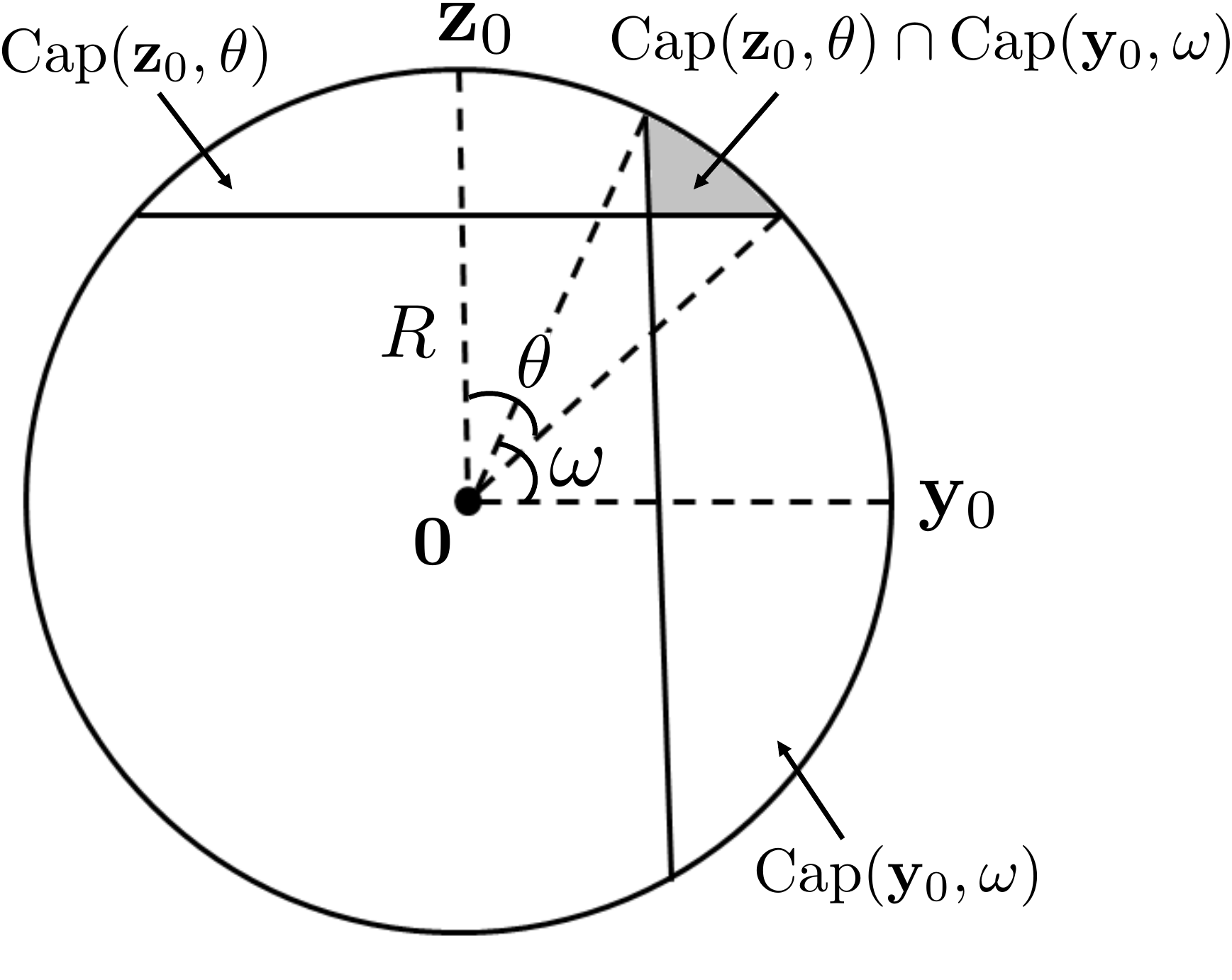}
\caption{Intersection of two spherical caps.}
\label{F:twocap}
\end{figure}

\section{Information Tension in \\ A Symmetric Markov Chain}\label{sec:inequality}

In this section, we prove an inequality between information measures in a certain type of Markov chain, which can be of interest in its own right. The proof of this inequality builds on Theorem~\ref{thm:strongisoperimetryshell}  from the previous section.  As we will see in Section \ref{S:proof},  the main theorems in this paper, i.e. Theorems~\ref{thm:maintheorem} and \ref{thm:maintheorem2}, are almost immediate given this result. We now state this result in the following theorem.

\begin{theorem}\label{L:upperboundlemma}
Consider a Markov chain $I_n-Z^n-X^n-Y^n$ where $X^n$, $Y^n$ and $Z^n$ are $n$-length random vectors and $I_n=f_n(Z^n)$ is a deterministic mapping of $Z^n$ to a set of integers. Assume moreover that $Z^n$ and $Y^n$ are i.i.d. white Gaussian vectors given $X^n$, i.e. $Z^n, Y^n \sim \mathcal{N}(X^n, N\, I_{n\times n})$ where $I_{n\times n}$ denotes the identity matrix, $E[\|X^n\|^2]= nP$, and $H(I_n|X^n)=-n\log \sin \theta_n$ for some $\theta_n\in[0,\pi/2]$.  Then the following inequality holds for any $n$,
\begin{align}
&H(I_n|Y^n)\nonumber \\
  &\leq n\,\cdot  \min_{\omega\in \left(\frac{\pi}{2}-\theta_n, \frac{\pi}{2}\right]}   \frac{1}{2}\log \left(\frac{4\text{sin} ^2\frac{\omega}{2}(P+N-N\text{sin}^2 \frac{\omega}{2})}{(P+N)(\text{sin} ^2 \theta_n - \cos^2 \omega)} \right).
\label{eq:entropyineq}
\end{align}
\end{theorem}

Note that $H(I_n|Y^n)$ is trivially lower bounded by $H(I_n|X^n)$ for any Markov chain $I_n-Z^n-X^n-Y^n$. The above theorem says that if $I_n-Z^n-X^n-Y^n$ satisfies the conditions of the theorem, then $H(I_n|Y^n)$ can also be upper bounded in terms of $H(I_n|X^n)$. In particular, 
it  provides an upper bound on $H(I_n|Y^n)$ in terms of $\theta_n=\arcsin 2^{- \frac{1}{n}H(I_n|X^n)}$. It can be easily verified that this upper bound on $H(I_n|Y^n)$ is decreasing with increasing $\theta_n$, or equivalently decreasing with decreasing $H(I_n|X^n)$, and implies that $H(I_n|Y^n)\rightarrow 0$ as $H(I_n|X^n)\rightarrow 0$.

We next turn to proving Theorem~\ref{L:upperboundlemma}. The reader who is interested in seeing how this theorem leads to Theorems~\ref{thm:maintheorem} and \ref{thm:maintheorem2}, without seeing its own proof, can jump to Section~\ref{S:proof}.  
In order to prove Theorem~\ref{L:upperboundlemma}, we will first establish some properties that are satisfied with high probability by long i.i.d. sequences generated from the source distribution $(I_n, Z^n, X^n, Y^n)$ satisfying the assumptions of the theorem. We now state and discuss these properties in Section~\ref{sec:typicality} and then use them to prove Theorem~\ref{L:upperboundlemma} in Section~\ref{SS:Proofoutline}. 

\subsection{Typicality Lemmas}\label{sec:typicality}
Assume $(I_n, Z^n, X^n,Y^n)$ satisfy the assumptions of Theorem~\ref{L:upperboundlemma}.   Consider the $B$-length i.i.d. sequence
\begin{align}
\{   (I_n(b), Z^n(b), X^n(b),Y^n(b) )    \}_{b=1}^{B},
\end{align}
where for any $b\in [1:B]$, $(I_n(b), Z^n(b), X^n(b),Y^n(b) )$ has the same distribution as $(I_n, Z^n, X^n,Y^n)$. For notational convenience,  in the sequel we write the $B$-length sequence $[X^n(1),X^n(2),\ldots,X^n(B)]$ as $\mathbf X$ and similarly define $\mathbf Y, \mathbf Z$ and $\mathbf I$; note that we have
$\mathbf I=[f_n(Z^n(1)), f_n(Z^n(2)),\ldots,f_n(Z^n(B)) ]=: f(\mathbf Z)$.  Also let $\text{Shell}\left( \mathbf c, r_1,r_2 \right) $  denote the spherical shell 
\begin{align*}
\text{Shell}\left( \mathbf c, r_1,r_2 \right):= \left\{ \mathbf a\in \mathbb R^{nB}: r_1\leq \|\mathbf a - \mathbf c \|\leq   r_2 \right\},
\end{align*}
and  let $\mbox{Ball}(\mathbf{c},r)$ denote the Euclidean ball 
\begin{align*}
\text{Ball}\left( \mathbf c,  r \right):= \left\{ \mathbf a\in \mathbb R^{nB}: \|\mathbf a - \mathbf c \|\leq   r \right\}.
\end{align*}
We next state several properties that $\mathbf X, \mathbf Y, \mathbf Z, \mathbf I$ satisfy with high probability when $B$ is large. The proofs of these properties are given in Appendix \ref{A:Proof_Typicality}.
\medbreak 
\begin{lemma}\label{L:E1}
For any $\delta>0$ and $B$ sufficiently large, we have
\begin{align*}
&\mbox{Pr}(E_1)\geq 1-\delta\\
 \text{ and \ \ }& \mbox{Pr}(E_2)\geq 1-\delta,
\end{align*}
where $E_1$ and $E_2$ are defined to be the following two events respectively: 
\begin{align}
& \left\{ \mathbf Z \in \text{Shell}\left( \mathbf 0, \sqrt{nB(P+N-\delta)},\sqrt{nB(P+N+\delta)}  \right)  \right\}, \label{E:defE1}
\end{align}
and
\begin{align}
 & \left\{ \mathbf Y \in \text{Shell}\left( \mathbf 0, \sqrt{nB(P+N-\delta)},\sqrt{nB(P+N+\delta)}  \right) \right\}. \label{E:defE2}
\end{align}
\end{lemma}


The proof of this lemma is a simple application of the law of large numbers and is included in Appendix \ref{SS:typicalityproof1}. The lemma simply states that when $B$ is large, $\mathbf Y$ and $\mathbf Z$ will concentrate in a thin $nB$-dimensional shell of radius $\sqrt{nB(P+N)}$.
\bigbreak
\begin{lemma}\label{L:prob_s(x,i)}
Given any $\epsilon>0$ and a pair of $(\mathbf x, \mathbf i)$, let $S_{\epsilon}(Z^n| \mathbf{x},\mathbf{i}  )$ be a set of $\mathbf z$'s defined as\footnote{Note that under this definition of $S_\epsilon(Z^n|\mathbf x, \mathbf i) $, if a pair $(\mathbf x, \mathbf i)$ doesn't satisfy
$2^{ nB(  \log \text{sin} \theta_n-\epsilon)}\leq p( \mathbf i|\mathbf x) \leq 2^{ nB( \log \text{sin} \theta_n+\epsilon)},$
then the set $S_\epsilon(Z^n|\mathbf x, \mathbf i) $ is empty because no $\mathbf z$ can satisfy the condition in \dref{E:thirdconditionS}.}
\begin{align}
  &   S_{\epsilon}(Z^n| \mathbf{x},\mathbf{i} )  :=\Big \{ \mathbf{z} \in f^{-1}(\mathbf i):\nonumber \\
    & ~~~~~~~  { \|\mathbf x-\mathbf z\| \in [\sqrt{nB(N-\epsilon)} ,  \sqrt{nB(N+\epsilon)}  ]}\\
 & ~~~~~~~\mathbf z \in \text{Ball}\left( \mathbf 0, \sqrt{nB(P+N+\epsilon)}  \right)\\
 &~~~~~~~  2^{ nB(  \log \text{sin} \theta_n-\epsilon)}\leq p( f(\mathbf{z})|\mathbf{x})\leq 2^{ nB( \log \text{sin} \theta_n+\epsilon)}
 \Big \}\label{E:thirdconditionS}
\end{align}
where $\theta_n=\arcsin 2^{- \frac{1}{n}H(I_n|X^n)}$ as in Theorem~\ref{L:upperboundlemma}. Then for $B$ sufficiently large, there exists a set $S_{\epsilon}(X^n,I_n)$ of $(\mathbf x, \mathbf i)$ pairs, such that
\begin{align}
\mbox{Pr}((\mathbf X, \mathbf I)\in S_{\epsilon}(X^n,I_n))\geq 1- \sqrt{\epsilon},\label{E:part1}
\end{align}
and for any $(\mathbf x, \mathbf i)\in S_{\epsilon}(X^n,I_n)$,
\begin{align}\mbox{Pr}(\mathbf Z \in S_{\epsilon}(Z^n| \mathbf{x},\mathbf{i})|\mathbf{x})\geq 2^{ nB(  \log \text{sin} \theta_n-2\epsilon)}.
\label{E:part2}
\end{align}
\end{lemma}

This lemma establishes the existence of a high probability set $S_{\epsilon}(X^n,I_n)$ of $(\mathbf x, \mathbf i)$ sequences, and a conditional typical set $S_{\epsilon}(Z^n| \mathbf{x},\mathbf{i})$ for each $(\mathbf x, \mathbf i)\in S_{\epsilon}(X^n,I_n)$ such that $ \mathbf z \in S_{\epsilon}(Z^n| \mathbf{x},\mathbf{i})$ satisfies some natural properties. Note that all properties  in the definition of  $S_{\epsilon}(Z^n| \mathbf{x},\mathbf{i})$ as well as \eqref{E:part2} are analogous to properties of strongly typical sets as stated in  \cite[Ch. 2]{ElGamalKim}. However, the notion of strong typicality does not apply to the current case since $Z^n$ and $Y^n$ are continuous random vectors and $X^n$ may or may not be continuous. Nevertheless, analogous properties can still be proved in this case; see the proof of this lemma in Appendix~\ref{SS:ProofStrongTyp}.
\medbreak

The following result has a slightly different flavor from the previous two lemmas in that it is simply a corollary of Theorem~\ref{thm:strongisoperimetryshell} from Section~\ref{sec:geometry}.

{\begin{corollary}\label{C:isoperi}
{For any $N,\epsilon$ such that $N>\epsilon>0$}, consider the spherical shell in $\mathbb R^m$
\begin{align*} &\mbox{Shell}\left(\mathbf 0,\sqrt{m(N-\epsilon)}  ,  \sqrt{m(N+\epsilon)}  \right)\\
&= \left\{\mathbf y\in \mathbb{R}^m : \; \sqrt{m(N-\epsilon)} \leq \|\mathbf y\| \leq \sqrt{m(N+\epsilon)}\right\}.
\end{align*}
Let $A \subseteq  \mbox{Shell}\left(\mathbf 0,\sqrt{m(N-\epsilon)}  ,  \sqrt{m(N+\epsilon)}  \right)$ be an arbitrary subset on this shell with volume 
\begin{equation}\label{eq:measureA}
|A| \geq 2^{ \frac{m}{2}\log 2\pi e (N+\epsilon)  \text{sin}^2 \theta},
\end{equation}
where $\theta\in (0,\pi/2)$. For any $\omega\in (\pi/2-\theta, \pi/2]$ 
and $m$ sufficiently large, we have 
\begin{align}
\text{Pr}\Bigg(& \left|A\cap \text{Ball}\left(\mathbf Y,  2\sqrt{m(N+\epsilon)} \sin \frac{\omega+\epsilon}{2} +2\sqrt{m\epsilon} \right)\right| \nonumber \\
&\geq 2^{\frac{m}{2}[\log(2\pi eN(\text{sin}^2 \theta   - \cos^2 \omega))-\epsilon]}\Bigg)\geq 1-\epsilon, \label{E:euctoshow}
\end{align}
where $\mathbf Y$ is drawn from any rotationally invariant 
distribution on the $\mbox{Shell}\left(\mathbf 0,\sqrt{m(N-\epsilon)}  ,  \sqrt{m(N+\epsilon)}  \right)$.
\end{corollary}}

This is a simple corollary of Theorem~\ref{thm:strongisoperimetryshell} when applied to a specific shell and a subset $A$ of this shell with measure prescribed by \eqref{eq:measureA}. The prescribed measure means that $A$ has an effective angle (asymptotically) greater than or equal to $\theta$. The corollary follows by observing that due to the triangle inequality (see also  Fig.~\ref{F:capinball}), for any $\mathbf y$ in the shell, 
$\text{ShellCap}(\mathbf y, \omega+\epsilon)$ considered in Theorem~\ref{thm:strongisoperimetryshell} is contained in the  Euclidean ball
$$ \text{Ball}\left(\mathbf y,  2\sqrt{m(N+\epsilon)} \sin \frac{\omega+\epsilon}{2}+ 2\sqrt{m\epsilon} \right).$$
The lower bound on the intersection volume in \dref{E:euctoshow} follows from an explicit characterization of 
$$V=|\text{ShellCap}(\mathbf z_0, \theta) \cap \text{ShellCap}(\mathbf y_0, \omega)|$$ in Theorem~\ref{thm:strongisoperimetryshell}, where $\angle(\mathbf z_0,\mathbf y_0)=\pi/2$ and $\theta+\omega>\pi/2$; see Appendix \ref{AS:twocap}, and in particular Lemma \ref{L:volumeintersectionshellcap}, for this characterization. A formal proof of  Corollary \ref{C:isoperi} is given in Appendix \ref{SS:ProofCor}.

\begin{figure}[t!]
\centering
\includegraphics[width=0.40\textwidth]{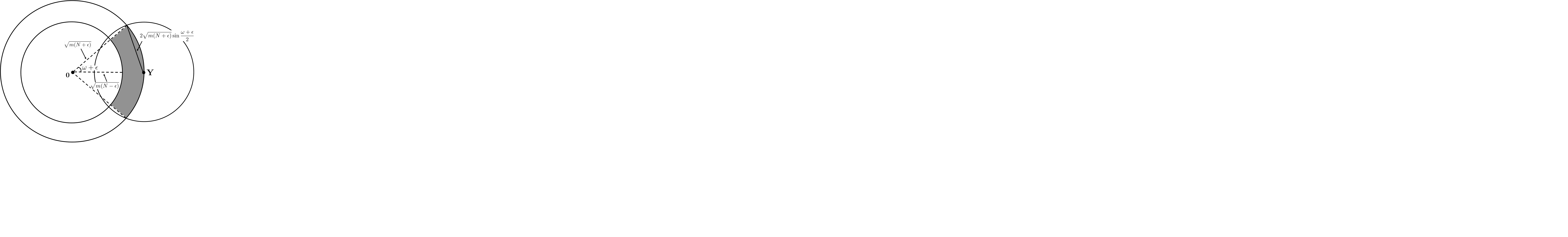}
\caption{Euclidean ball contains the shell cap.}
\label{F:capinball}
\end{figure}

\medbreak

The above corollary together with Lemma~\ref{L:prob_s(x,i)} leads to the following lemma.
 
\begin{lemma}\label{L:Keylemma}
For any $\delta>0$ and $B$ sufficiently large, we have
\begin{align*}
\mbox{Pr}(E_3)\geq 1-\delta,
\end{align*}
where $E_3$ is defined to be the following event:
\begin{align}
\Bigg\{ \Bigg| & f^{-1}(\mathbf I)\cap \text{Ball} \left(\mathbf 0, \sqrt{nB  (P+N +\delta}) \right) \nonumber \\
& \cap  \text{Ball} \left(\mathbf Y, \sqrt{nB N  \left( 4\sin^2 \frac{\omega}{2} + \delta \right)}  \right)       \Bigg| \nonumber \\
 &~~~~~~~~~~~~\geq   2^{nB[\frac{1}{2}\log(2\pi eN(\text{sin} ^2 \theta_n - \cos^2 \omega))-\delta]}     \Bigg \}  \label{E:defE3}
\end{align}
in which $f^{-1}(\mathbf I):=\{\mathbf a \in \mathbb R^{nB}: f(\mathbf a)=\mathbf I \}$ and {$\omega\in (\pi/2-\theta_n+\delta,\pi/2 ]$}.
\end{lemma}
 
This  lemma can also be regarded as a typicality lemma as it states a property satisfied  by $(\mathbf I, \mathbf Y)$ pair with high probability  when $B$ is large. However, this is a non-trivial property. The lemma follows by first fixing a pair $(\mathbf x, \mathbf i)\in S_{\epsilon}(X^n,I_n)$ and showing that the volume of the set $S_{\epsilon}(Z^n| \mathbf{x},\mathbf{i})$ defined in Lemma~\ref{L:prob_s(x,i)} can be lower bounded by
$$2^{\frac{nB}{2}\log(2\pi eN\text{sin} ^2 \theta_n )},$$
up to the first order term in the exponent. Since by definition $S_{\epsilon}(Z^n| \mathbf{x},\mathbf{i})$ is a subset of the shell
$$
\mbox{Shell}\left(\mathbf{x},\sqrt{nB(N-\epsilon)}  ,  \sqrt{nB(N+\epsilon)}  \right), $$
and given $\mathbf X= \mathbf x$, $\mathbf Y$ is isotropic Gaussian (therefore rotationally invariant around $\mathbf x$ when constrained to this shell), we can apply Corollary~\ref{C:isoperi} to the above shell by choosing the set $A$  to be $S_{\epsilon}(Z^n| \mathbf{x},\mathbf{i})$. This allows us to conclude that
\begin{align}\label{eq:intermediate}
&\mbox{Pr}\Bigg( \left|S_{\epsilon}(Z^n| \mathbf{x},\mathbf{i})\cap \text{Ball}\left(\mathbf Y, \sqrt{nB N  \bigg( 4\sin^2 \frac{\omega}{2} + \epsilon \bigg)}\right) \right| \nonumber \\
&~~~~~~~ \geq 2^{nB\left[\frac{1}{2}\log(2\pi eN(\text{sin} ^2 \theta_n - \cos^2 \omega))-\epsilon \right]} \Bigg| \mathbf X =  \mathbf x\Bigg)\geq 1-\epsilon.
\end{align}
The conclusion of Lemma~\ref{L:Keylemma} then follows by observing that by definition
$$S_{\epsilon}(Z^n| \mathbf{x},\mathbf{i})\subseteq f^{-1}(\mathbf i) \cap \text{Ball}\left( \mathbf 0, \sqrt{nB(P+N+\epsilon)}  \right),$$ and removing the conditioning with respect to $\mathbf X$ in \eqref{eq:intermediate}. The formal proof of Lemma \ref{L:Keylemma} is given in Appendix~\ref{SS:appendixproofkeylemma}.

\subsection{Proof of Theorem~\ref{L:upperboundlemma}}\label{SS:Proofoutline}

We are now ready to prove Theorem~\ref{L:upperboundlemma}, which mainly builds on Lemma~\ref{L:Keylemma}. Consider a $\mathbf Y$ that with high probability lies in the ball with center $\mathbf 0$ and approximate radius $\sqrt{nB(P+N)}$, and draw another ball around $\mathbf Y$ of approximate radius  $\sqrt{nB N   4\sin^2 \frac{\omega}{2}   }$ and intersect this ball with the original ball; equivalently, this corresponds to considering a cap around $\mathbf{Y}$ of angle $\phi$ on the original ball (see Fig.~\ref{F:twoball}). Lemma~\ref{L:Keylemma} asserts that this cap around $\mathbf{Y}$ will have a certain minimal intersection volume with $f^{-1}(\mathbf I)$. In other words, there is a subset of this cap with certain minimal volume that is mapped to $\mathbf I$. This naturally lends itself to a packing argument: the number of distinct $\mathbf I$ values plausible under a given $\mathbf Y$ can be upper bounded by the ratio between the volume of the cap around $\mathbf Y$ and the minimal intersection volume occupied for each distinct $\mathbf I$. This in turn leads to a bound on $H(\mathbf I|\mathbf Y)$.

\begin{figure}[t!]
\centering
\includegraphics[width=0.40\textwidth]{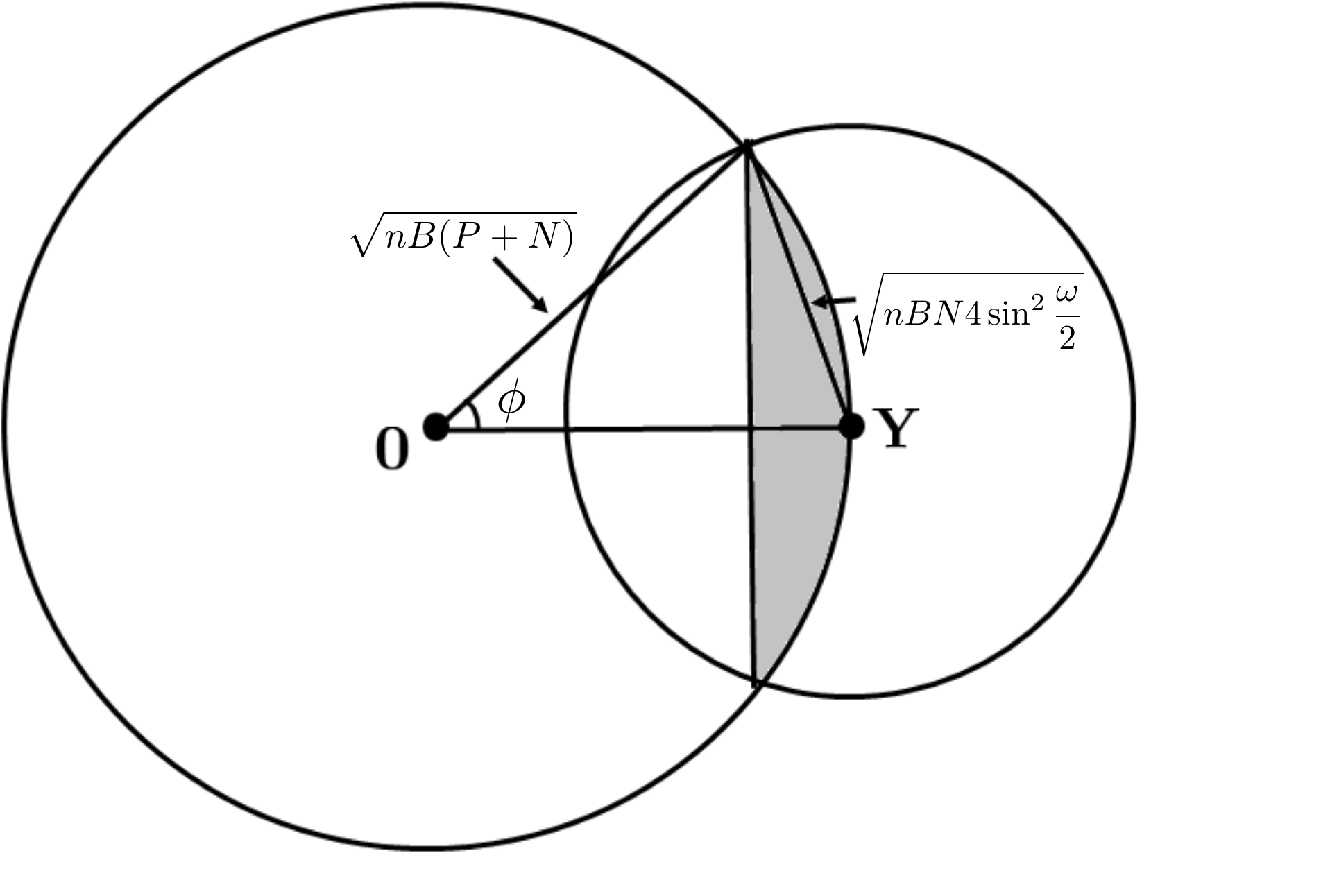}
\caption{A spherical cap with angle $\phi=2\arcsin\sqrt{\frac{N\text{sin}^2\frac{\omega }{2}}{P+N}}$.}
\label{F:twoball}
\end{figure}

We now proceed with the formal proof. Consider
the indicator function
$$F=\mathbb I (E_1,E_2,E_3) $$
where $\mathbb{I}(\cdot)$ is defined as
\begin{numcases}{\mathbb{I}(A)=}
1  \text{~~~if $A$ holds} \nonumber \\
0  \text{~~~otherwise, } \nonumber
\end{numcases}
and the  events $E_1,E_2$ and $E_3$ are as given by \dref{E:defE1}, \dref{E:defE2} and \dref{E:defE3} respectively.
Obviously, by the union bound, we have
$$\mbox{Pr}(F=1)\geq 1-3\delta$$
for any $\delta>0$ and $B$ sufficiently large, and therefore
\begin{align}
 H(\mathbf I |\mathbf Y)
& \leq H(\mathbf I , F|\mathbf Y)\nonumber \\
&= H( F|\mathbf Y)+H(\mathbf I |\mathbf Y,F)\nonumber \\
&\leq H(\mathbf I |\mathbf Y,F)+1\nonumber \\
&=\mbox{Pr}(F=1) H(\mathbf I |\mathbf Y,F=1)\nonumber \\
&~~~+\mbox{Pr}(F=0) H(\mathbf I |\mathbf Y,F=0)+ 1\nonumber \\
&\leq H(\mathbf I |\mathbf Y,F=1) +3\delta nBC_0 +1  .\label{E:plug}
\end{align}

{To bound $H(\mathbf I |\mathbf Y,F=1) $, it suffices to bound $H(\mathbf I |\mathbf Y=\mathbf y, F=1)$ for any
\begin{align}
\mathbf y \in \text{Shell}\left( \mathbf 0, \sqrt{nB(P+N-\delta)},\sqrt{nB(P+N+\delta)}  \right). \label{E:ensure}
\end{align}
For this, we apply a packing argument as follows. Consider a ball centered at any $\mathbf y$ satisfying \dref{E:ensure} and of radius $\sqrt{nB N  \left( 4\sin^2 \frac{\omega}{2} + \delta \right)}$, i.e.,
$$\text{Ball}  \left( \mathbf y, \sqrt{nB N  \left( 4\sin^2 \frac{\omega}{2} + \delta \right)}  \right) ,$$
where $\omega$ satisfies
$$\pi/2-\theta_n+\delta<\omega\leq \pi/2.$$We now use the following lemma (whose proof is included in Appendix \ref{A:twoball}) to upper bound the volume of the intersection between this ball and $\text{Ball}\left( \mathbf 0, \sqrt{nB(P+N+\delta)}  \right) $, i.e.,
\begin{align*}
&\Bigg|\text{Ball}\left( \mathbf y, \sqrt{nB N  \left( 4\sin^2 \frac{\omega}{2} + \delta \right)} \right) \\
&~~~~~~~~~ \cap \text{Ball}\left( \mathbf 0, \sqrt{nB(P+N+\delta)}  \right) \Bigg|.
\end{align*}
\begin{lemma}\label{L:twoball}
Let $\text{Ball}(\mathbf c_1, \sqrt{mR_1})$ and $\text{Ball}(\mathbf c_2, \sqrt{mR_1})$ be two balls in $\mathbb R^m$ with $\|\mathbf c_1 -\mathbf c_2\|=\sqrt{mD}$, where $D$ satisfies $ (\sqrt{R_1}-\sqrt{R_2})^2< D< (\sqrt{R_1}+\sqrt{R_2})^2$. Then for any $\epsilon>0$ and $m$ sufficiently large, we have
\begin{align*}
&\left|\text{Ball}(\mathbf c_1, \sqrt{mR_1})\cap \text{Ball}(\mathbf c_2, \sqrt{mR_1})\right|\\
&\leq  2^{m\left(\frac{1}{2}\log \pi e \lambda(R_1,R_2,D)    +\epsilon \right)}
\end{align*}
where
\begin{align*}
\lambda(R_1,R_2,D):=\frac{ 2R_1D+2R_1R_2+2DR_2-R_1^2-R_2^2-D^2 }{2D} .
\end{align*}
\end{lemma}

Using the above lemma, we have for $B$ sufficiently large,
\begin{align*}
&\Bigg|\text{Ball}\left( \mathbf y, \sqrt{nB N  \left( 4\sin^2 \frac{\omega}{2} + \delta \right)} \right) \\
&~~~~~~~~ \cap \text{Ball}\left( \mathbf 0, \sqrt{nB(P+N+\delta)}  \right) \Bigg|\\
&\leq 2^{nB\left[\frac{1}{2}\log \pi e \lambda\left(P+N+\delta, N \left( 4\text{sin}^2 \frac{\omega}{2} + \delta\right)  ,\|\mathbf y \|\right)    +\delta \right]}\\
&= 2^{nB\left[\frac{1}{2}\log \pi e \lambda\left(P+N,   4N\text{sin}^2 \frac{\omega}{2}    ,P+N \right)    +\delta_1 \right]} \\
&=2^{ nB \left[\frac{1}{2}\log \frac{8\pi e N\text{sin}^2 \frac{\omega}{2}(P+N-N\text{sin}^2 \frac{\omega}{2})}{P+N}     +\delta_1\right]},
\end{align*}
for some $\delta_1\to 0$ as $\delta \to 0$, where the first inequality is an immediate application of Lemma \ref{L:twoball}, the first equality follows from the fact that 
$$\mathbf y \in \text{Shell}\left( \mathbf 0, \sqrt{nB(P+N-\delta)},\sqrt{nB(P+N+\delta)}  \right)$$ 
and the continuity of the function $\lambda(R_1,R_2,D)$ in its arguments, and the second equality follows from a simple evaluation of $\lambda\left(P+N,   4N\text{sin}^2 \frac{\omega}{2}    ,P+N \right)$.

On the other hand,  the condition $F=1$ (c.f. the definition of $E_3$ in Lemma \ref{L:Keylemma}) also ensures that
\begin{align*}
&\Bigg|  f^{-1}(\mathbf I)\cap \text{Ball} \left(\mathbf 0, \sqrt{nB  (P+N +\delta} \right)
 \\  
&~~~~~~ \cap  \text{Ball} \left(\mathbf y, \sqrt{nB N  \left( 4\sin^2 \frac{\omega}{2} + \delta \right)}  \right)       \Bigg|\\
& \geq   2^{nB[\frac{1}{2}\log(2\pi eN(\text{sin} ^2 \theta_n - \cos^2 \omega))-\delta]}.
\end{align*}
Since $f^{-1}(\mathbf i)$ are disjoint sets for different $\mathbf i$, given $F=1$ and $\mathbf{Y}=\mathbf{y}$, 
the number of different possible values for $\mathbf I$  can be upper bounded by the ratio between
\begin{align*}
&\Bigg|\text{Ball}\left( \mathbf y, \sqrt{nB N  \left( 4\sin^2 \frac{\omega}{2} + \delta \right)} \right) \\
&~~~~~~~~~ \cap \text{Ball}\left( \mathbf 0, \sqrt{nB(P+N+\delta)}  \right) \Bigg|
\end{align*}
and}
$$2^{nB[\frac{1}{2}\log(2\pi eN(\text{sin} ^2 \theta_n - \cos^2 \omega))-\delta]},$$
which can be further upper bounded by
\begin{align*}
&2^{nB\left[\frac{1}{2}\log \frac{8\pi e N\text{sin} ^2\frac{\omega}{2}(P+N-N\text{sin}^2 \frac{\omega}{2})}{P+N} -\frac{1}{2}\log(2\pi eN(\text{sin} ^2 \theta_n - \cos^2 \omega)) +\delta+\delta_1   \right]} \\
& =2^{nB\left[\frac{1}{2}\log \frac{4\text{sin} ^2\frac{\omega}{2}(P+N-N\text{sin}^2 \frac{\omega}{2})}{(P+N)(\text{sin} ^2 \theta_n - \cos^2 \omega)}   + \delta_2     \right]} ,
\end{align*}
where $\delta_2\to 0$ as $\delta\to 0$.
This immediately implies the following upper bound on $H(\mathbf I |\mathbf Y= \mathbf y,F=1) $ and therefore $H(\mathbf I |\mathbf Y,F=1) $,
\begin{align*}
&H(\mathbf I |\mathbf Y,F=1)\\
& \leq nB\left[\frac{1}{2}\log \frac{4\text{sin} ^2\frac{\omega}{2}(P+N-N\text{sin} ^2\frac{\omega}{2})}{(P+N)(\text{sin} ^2 \theta_n - \cos^2 \omega)}  +\delta_2   \right],\nonumber
\end{align*}
which combined with \dref{E:plug} yields that
\begin{align*}
H(\mathbf I |\mathbf Y) \leq \ &  nB \left[\frac{1}{2}\log \frac{4\text{sin}^2 \frac{\omega}{2}(P+N-N\text{sin}^2 \frac{\omega}{2})}{(P+N)(\text{sin} ^2 \theta_n - \cos^2 \omega)}   +\delta_2   \right] \\
& +3\delta nBC_0 +1.
\end{align*}
Dividing both sides of the above inequality by $B$ and noting that
$$H(\mathbf I |\mathbf Y)=\sum_{b=1}^B H(I_n(b)|Y^n(b))=BH(I_n|Y^n),$$
we have
\begin{align}
&H(I_n|Y^n)\nonumber \\
&  \leq  n\bigg( \frac{1}{2}\log \frac{4\text{sin}^2 \frac{\omega}{2}(P+N-N\text{sin}^2 \frac{\omega}{2})}{(P+N)(\text{sin} ^2 \theta_n - \cos^2 \omega)}  +\delta_2+3\delta C_0+\frac{1}{nB}\bigg),\label{E:arbitrarysmall1}
\end{align}
{which holds for any
\begin{align}
\omega\in (\pi/2-\theta_n+\delta, \pi/2].\label{E:arbitrarysmall2}
\end{align}
Since $\delta, \delta_2$ and $\frac{1}{nB}$ in \dref{E:arbitrarysmall1}--\dref{E:arbitrarysmall2} can all be made arbitrarily small by choosing $B$ sufficiently large, we obtain
\begin{align}
H(I_n|Y^n) \leq  n \left(  \frac{1}{2}\log \frac{4\text{sin} ^2\frac{\omega}{2}(P+N-N\text{sin}^2 \frac{\omega}{2})}{(P+N)(\text{sin} ^2 \theta_n - \cos^2 \omega)}    \right), \label{E:omegatrans}
\end{align}
for any $\omega\in \left(\frac{\pi}{2}-\theta_n, \frac{\pi}{2}\right]$. This completes the proof of Theorem \ref{L:upperboundlemma}.}

\section{Proofs of Theorems~\ref{thm:maintheorem} and \ref{thm:maintheorem2}}\label{S:proof}

We now prove Theorem~\ref{thm:maintheorem2} by using Theorem~\ref{L:upperboundlemma}, and use Theorem~\ref{thm:maintheorem2}  to prove Theorem~\ref{thm:maintheorem}.

\subsection{Proof of Theorem~\ref{thm:maintheorem2}}
Suppose a rate $R$ is achievable. Then there exists a sequence of $(2^{nR},n)$ codes
such that the average probability of error $P_e^{(n)} \to 0$ as $n \to \infty$.
Let the relay's transmission be denoted by $I_n=f_n(Z^n)$. By standard information theoretic arguments, for this sequence of codes we have
\begin{align}
n R &=  H(M)\nonumber \\
&=I(M;Y^n,I_n)+H(M|Y^n,I_n)\nonumber \\
&\leq I (X^n;Y^n,I_n)+n\mu \label{eq:Fano}  \\
&=I(X^n;Y^n)+I(X^n;I_n|Y^n)+ n\mu  \nonumber \\
&=I(X^n;Y^n)+H(I_n|Y^n)-H(I_n|X^n)+ n\mu  \label{E:markov}\\
&\leq nI(X_Q;Y_Q)+H(I_n|Y^n)-H(I_n|X^n)+ n\mu \label{E:Timesharing}\\
&\leq \frac{n}{2}\log \left(1+\frac{P}{N}\right)+H(I_n|Y^n)-H(I_n|X^n)+ n\mu, \label{E:power}
\end{align}
for any $\mu>0$ and $n$ sufficiently large. In the above, \dref{eq:Fano} follows from applying the data processing inequality to the Markov chain $M-X^n-(Y^n,I_n)$ and Fano's inequality, \dref{E:markov} uses the fact that $I_n-X^n-Y^n$ form a Markov chain and thus $H(I_n|X^n,Y^n)=H(I_n|X^n)$, \dref{E:Timesharing} follows by defining the time sharing random variable $Q$ to be uniformly distributed over $[1:n]$, and \dref{E:power} follows because
\begin{align*} 
E[X^2_Q]&=  \frac{1}{{2^{nR}}}\sum_{m=1}^{2^{nR}}\frac{1}{n}\sum_{i=1}^{n} x^2_i(m)  \\
&=  \frac{1}{n} \frac{1}{{2^{nR}}} \sum_{m=1}^{{2^{nR}}}\|x^n(m)\|^2  \\
&\leq P.
\end{align*}
Given \eqref{E:power}, the standard way to proceed would be to upper bound the first entropy term by $H(I_n|Y^n)\leq H(I_n)\leq nC_0$ and lower bound the second entropy term $H(I_n|X^n)$ simply by $0$. This would lead to the so-called multiple-access bound in the well-known cut-set bound on the capacity of this channel \cite{covelg79}. However, as we already point out in our previous works \cite{IZS2016}--\cite{ISIT2016}, this leads to a loose bound since it does not capture the inherent tension between how large the first entropy term can be and how small the second one can be. Instead, we can use Theorem~\ref{L:upperboundlemma} to more tightly  upper bound the difference $H(I_n|Y^n)-H(I_n|X^n)$ in \dref{E:power}.

We start by verifying that the random variables $I_n, X^n, Z^n$ and $Y^n$ associated with a code of blocklength $n$ satisfy the conditions in Theorem~\ref{L:upperboundlemma}. It is trivial to observe that they satisfy the required Markov chain condition and $Z^n$ and $Y^n$ are i.i.d. Gaussian given $X^n$ due to the channel structure. Also assume that  
$$
E[\|X^n\|^2]= \frac{1}{{2^{nR}}}\sum_{m=1}^{2^{nR}}\|x^n(m)\|^2= nP'
$$
with $P'\leq P$, and assume that $H(I_n|X^n)=-n\log \sin \theta_n$.  Then, applying Theorem~\ref{L:upperboundlemma} to the random variables associated with a code for the relay channel, we have
\begin{align*}
&H(I_n|Y^n) \nonumber \\
&\leq   n\,\cdot  \min_{\omega\in \left(\frac{\pi}{2}-\theta_n, \frac{\pi}{2}\right]}   \frac{1}{2}\log \left(\frac{4\text{sin} ^2\frac{\omega}{2}(P'+N-N\text{sin}^2 \frac{\omega}{2})}{(P'+N)(\text{sin} ^2 \theta_n - \cos^2 \omega)} \right)\\
&\leq  n\,\cdot  \min_{\omega\in \left(\frac{\pi}{2}-\theta_n, \frac{\pi}{2}\right]}   \frac{1}{2}\log \left(\frac{4\text{sin} ^2\frac{\omega}{2}(P +N-N\text{sin}^2 \frac{\omega}{2})}{(P +N)(\text{sin} ^2 \theta_n - \cos^2 \omega)} \right),
\end{align*}
and therefore, 
\begin{align}
H(I_n|Y^n) -H(I_n|X^n) &\leq   n\,\cdot  \min_{\omega\in \left(\frac{\pi}{2}-\theta_n, \frac{\pi}{2}\right]}  h_{\theta_n}(\omega) \label{E:plugdiffer}
\end{align}
where $h_{\theta_n}(\omega)$ is defined as
\begin{align}&~h_{\theta_n}(\omega)=    \frac{1}{2}\log \left(\frac{4\text{sin} ^2\frac{\omega}{2}(P+N-N\text{sin}^2 \frac{\omega}{2})\sin^2\theta_n}{(P+N)(\text{sin} ^2 \theta_n - \cos^2 \omega)} \right),\label{E:hdef}
 \end{align}
in which $\theta_n =\arcsin 2^{- \frac{1}{n}H(I_n|X^n)}$  satisfies
 \begin{align}\label{thetanlb}
\theta_0:=\arcsin (2^{- C_0})\leq\arcsin 2^{- \frac{1}{n}H(I_n|X^n)}=\theta_n\leq \frac{\pi}{2}.
\end{align}
Plugging \dref{E:plugdiffer} into \dref{E:power}, we conclude that for any achievable rate $R$,
\begin{align}\label{eq:upperboundR}
R
&\leq \frac{1}{2}\log \left(1+\frac{P}{N}\right)+ \min_{\omega\in \left(\frac{\pi}{2}-\theta_n, \frac{\pi}{2}\right]}  h_{\theta_n}(\omega)+ \mu.
\end{align}

At the same time, for any achievable rate $R$, we also have
\begin{equation}\label{eq:upperboundR2}
R\leq \frac{1}{2}\log \left(1+\frac{P}{N}\right)+C_0+\log \sin \theta_n+\mu,
\end{equation}
which simply follows from \eqref{E:power} by upper bounding $H(I_n|Y^n)$ with $nC_0$ and plugging in the definition of $\theta_n$. Therefore, if a rate $R$ is achievable, then for any $\mu>0$ and $n$ sufficiently large it should simultaneously satisfy both \eqref{eq:upperboundR} and \eqref{eq:upperboundR2}  for some $\theta_n$ that satisfies the condition in \eqref{thetanlb}. This concludes the proof of the theorem.

\subsection{Proof of Theorem~\ref{thm:maintheorem}}
In order to show that Theorem~\ref{thm:maintheorem} follows from Theorem~\ref{thm:maintheorem2}, consider the following bound on $C(C_0)$ implied by Theorem~\ref{thm:maintheorem2}:
\begin{align}
C(C_0)  \leq & \ \frac{1}{2}\log \left(1+\frac{P}{N}\right)\nonumber \\
& ~~~~ + \sup_{\theta\in\left[\arcsin (2^{-C_0}),\frac{\pi}{2}\right]} \min_{\omega\in \left(\frac{\pi}{2}-\theta, \frac{\pi}{2}\right]}  h_{\theta}(\omega).\label{thm2:bound2}
\end{align}
With $\theta_0$ defined as $\arcsin (2^{- C_0})$, we can upper bound the right-hand side of \eqref{thm2:bound2} to obtain
\begin{align*}
C(C_0)
&\leq \frac{1}{2}\log \left(1+\frac{P}{N}\right)+ \sup_{\theta\in\left[\theta_0,\frac{\pi}{2}\right]}  \min_{\omega\in \left(\frac{\pi}{2}-\theta_0, \frac{\pi}{2}\right]}  h_{\theta}(\omega).
\end{align*}
Also because given any fixed $\omega\in \left(\frac{\pi}{2}-\theta_0, \frac{\pi}{2}\right]$, $h_{\theta}(\omega)\leq h_{\theta_0}(\omega)$ for any $\theta\in [\theta_0, \pi/2]$,  we further have
\begin{align}
C(C_0)
&\leq \frac{1}{2}\log \left(1+\frac{P}{N}\right)+ \min_{\omega\in \left(\frac{\pi}{2}-\theta_0, \frac{\pi}{2}\right]}  h_{\theta_0}(\omega).   \label{E:divide}
\end{align}
The significance of the function $h_{\theta_0}(\omega)$ is that for any $\theta_0>0$,
\begin{equation}\label{eq:h(0)}
h_{\theta_0 }\left(\frac{\pi}{2}\right)=\frac{1}{2}\log \left( \frac{2P+N}{P+N} \right),
\end{equation}
 and $h_{\theta_0 }(\omega)$ is increasing at $\omega=\frac{\pi}{2}$, or more precisely,
$$h'_{\theta_0 }\left(\frac{\pi}{2}\right)= \frac{P }{(2P+N)\ln 2 }>0.$$ Therefore, as long as $\theta_{0}>0$, which is the case when $C_0$ is finite,
the minimization of $h_{\theta_0}(\omega)$ with respect to $\omega$ in \eqref{E:divide} yields a value strictly smaller than $h_{\theta_0}\left(\frac{\pi}{2}\right)$ in \eqref{eq:h(0)}. This would allow us to conclude that the capacity $C(C_0)$ for any finite $C_0$ is strictly smaller than $\frac{1}{2}\log\left(1+\frac{2P}{N}\right)$.

We now formalize the above argument. Using the definition of the derivative, one obtains
$$h'_{\theta_0 }\left(\frac{\pi}{2}\right)=\lim_{\Delta\to 0} \frac{h_{\theta_0 }\left(\frac{\pi}{2}\right)-h_{\theta_0 }\left(\frac{\pi}{2}-\Delta\right)}{\Delta}.$$
Therefore, there exists a sufficiently small $\Delta_1>0$ such that
$0<\Delta_1<\theta_0$  and
$$\left| \frac{h_{\theta_0 }\left(\frac{\pi}{2}\right)-h_{\theta_0 }\left(\frac{\pi}{2}-\Delta_1\right)}{\Delta_1}-h'_{\theta_0 }\left(\frac{\pi}{2}\right)\right| \leq  \frac{h'_{\theta_0 }\left(\frac{\pi}{2}\right)}{2}.$$
For such $\Delta_1$ we have
\begin{align*}
h_{\theta_0 }\left(\frac{\pi}{2}-\Delta_1\right)&\leq h_{\theta_0 }\left(\frac{\pi}{2}\right)- \frac{\Delta_1 h'_{\theta_0 }\left(\frac{\pi}{2}\right)}{2}\\
&=\frac{1}{2}\log \left( \frac{2P+N}{P+N} \right)-\frac{P\Delta_1 }{2(2P+N)\ln 2 } ,
\end{align*}
which further implies that
\begin{align}
\min_{\omega\in \left(\frac{\pi}{2}-\theta_0, \frac{\pi}{2}\right]}  h_{\theta_0}(\omega)\leq \frac{1}{2}\log \left( \frac{2P+N}{P+N} \right)-\frac{P\Delta_1 }{2(2P+N)\ln 2 } . \label{E:combo2}
\end{align}
Combining \dref{E:divide} and \dref{E:combo2} we obtain that for any finite $C_0$, there exists some $\Delta_1>0$ such that
\begin{align}
C(C_0)\leq \frac{1}{2}\log \left( 1+\frac{2P}{N} \right)-\frac{P\Delta_1 }{2(2P+N)\ln 2 }.
\end{align}
This proves Theorem~\ref{thm:maintheorem}.

\section{Conclusion}

We have proved a new upper bound on the capacity of the Gaussian relay channel and solved a problem  posed by Cover in \cite{cg87}, which has remained open since 1987. The derivation of our upper bound focuses on directly characterizing the tension between information measures of pertinent $n$-letter random variables. In particular, this is done via the following steps:
\begin{itemize}
\item we first use ``typicality'' to translate the information tension problem to a problem regarding the geometry of the typical sets of these $n$-letter random variables;
\item we then use results and tools in the (broadly defined) field of concentration of measure, in particular rearrangement theory, to establish non-trivial geometric properties for these typical sets;
\item we finally use these geometric properties to construct a packing argument, which leads to an inequality between the original $n$-letter information measures.
\end{itemize} 

In contrast, the typical program for proving converses in network information theory focuses on ``single-letterizing'' $n$-letter information measures. This makes it difficult to invoke tools from geometry and concentration of measure, which in retrospect appear well-suited for quantifying information tensions that lie at the hearth of network problems. Indeed, to the best of our knowledge, the use of concentration of measure in information theory has been mostly limited to establishing strong converses for problems whose capacity is already known (c.f., e.g. \cite{Csiszarbook, mono}), and it has been rarely used to derive first-order results, i.e. bounds on the capacity of multi-user networks. Our proof suggests that measure concentration, in particular geometric inequalities and their functional counterparts, can have a bigger role to play in network information theory. It would be interesting to better understand this role and see if the program developed in this paper can be used to prove converses for other open problems in network information theory.

\appendices

\section{Proofs of Extended Isoperimetric Inequalities}\label{A:proofisop}
In this appendix, we prove the extended isoperimetric inequalities on the sphere and on the shell, as stated in Theorems \ref{thm:strongisoperimetrysphere} and \ref{thm:strongisoperimetryshell} respectively. In particular, we will first prove the shell case and then show that the sphere case follows as a corollary.

\subsection{Preliminaries}
We begin with some preliminaries that will be used in the proofs. 
Our main tool for proving Theorems \ref{thm:strongisoperimetrysphere} and \ref{thm:strongisoperimetryshell}  is the symmetric decreasing rearrangement of functions on the sphere, along with a version of the Riesz rearrangement inequality on the sphere due to Baernstein and Taylor \cite{baernstein}.

For any measurable function $f:\mathbb{S}^{m-1} \to \mathbb{R}$ and pole $\mathbf z_0$, the symmetric decreasing rearrangement of $f$ about $\mathbf z_0$ is defined to be the function $f^*:\mathbb{S}^{m-1} \to \mathbb{R}$ such that $f^*(\mathbf y)$ depends only on the angle $\angle(\mathbf y,\mathbf z_0)$, is nonincreasing in $\angle (\mathbf y,\mathbf z_0)$, and has super-level sets of the same Haar measure as $f$, i.e.

$$\mu\big(\{\mathbf y:f^*(\mathbf y)>d\}\big) = \mu\big(\{\mathbf y:f(\mathbf y)>d\}\big)$$
for all $d$. The function $f^*$ is unique up to its value on sets of measure zero.

One important special case is when the function $f=1_A$ is the characteristic function for a subset $A$. The function $1_A$ is just the function such that
$$1_A(\mathbf y) = \begin{cases} 1 & \mathbf y \in A \\ 0 &  \text{otherwise.} \end{cases} $$
In this case, $1_A^*$ is equal to the characteristic function associated with a spherical cap of the same size as $A$. In other words, if $A^*$ is a spherical cap about the pole $\mathbf z_0$ such that $\mu(A^*) = \mu(A)$, then $1_A^* = 1_{A^*}$.

\begin{lemma}[Baernstein and Taylor \cite{baernstein}]\label{thm:riesz}
Let $K$ be a nondecreasing bounded measurable function on the interval $[-1,1]$. Then for all functions $f,g \in L^1(\mathbb{S}^{m-1})$,

\begin{eqnarray*}
& &\int_{\mathbb{S}^{m-1}} \left( \int_{\mathbb{S}^{m-1}} f(\mathbf z)K\left(\langle \mathbf z/R,\mathbf y /R\rangle\right) d\mathbf z \right) g(\mathbf y)d\mathbf{y} \\
& & \leq \int_{\mathbb{S}^{m-1}} \left( \int_{\mathbb{S}^{m-1}} f^*(\mathbf z)K\left(\langle \mathbf z/R, \mathbf y/R \rangle\right) d\mathbf z \right) g^*(\mathbf y)d\mathbf{y}.
\end{eqnarray*}
\end{lemma}

{

For any $f\in L^1(\mathbb{S}^{m-1})$, define
$$\psi(\mathbf y) = \int_{\mathbb{S}^{m-1}} f(\mathbf z)K\left(\langle \mathbf z/R,\mathbf y /R\rangle\right) d\mathbf z$$
to be the inner integral in Lemma \ref{thm:riesz}. When applying Lemma \ref{thm:riesz} we will use test functions $g$ that are characteristic functions. Let $g=1_C$ where $C = \{\mathbf y:\psi(\mathbf y) > d\}$ for some $d$ (i.e. $C$ is a super-level set of $\psi$). For a fixed measure $\mu(C)$, the left-hand side of the inequality from Lemma \ref{thm:riesz} will be maximized by this choice of $C$. With this choice we have the following equality:
\begin{eqnarray*}
\int_{\mathbb{S}^{m-1}} \psi(\mathbf y) 1_C(\mathbf y)d\mathbf{y} & = & \int_{\mathbb{S}^{m-1}} \psi^*(\mathbf y) 1^*_C(\mathbf y)d\mathbf{y} \\
& = & \int_{C^*} \psi^*(\mathbf{y})d\mathbf {y} .
\end{eqnarray*}
This follows from the layer-cake decomposition for any non-negative and measurable function $\psi$ in that
\begin{align}
\int_{\mathbb{S}^{m-1}} \psi(\mathbf{y}) 1_C(\mathbf{y})d\mathbf{y} & = \int_C \psi(\mathbf{y})d\mathbf{y} \nonumber \\
& = \int_C \int_0^\infty 1_{\{\psi(\mathbf{y})>t\}}dt d\mathbf{y} \nonumber\\
& = \int_0^\infty \int_C 1_{\{\psi(\mathbf y)>t\}}d\mathbf{y} dt \nonumber\\
& = \int_0^\infty \int_{\mathbb{S}^{m-1}} 1_{\{\psi(\mathbf y)>\max(t,d)\}}d\mathbf{y} dt \nonumber\\
& = \int_0^\infty \int_{\mathbb{S}^{m-1}} 1_{\{\psi^*(\mathbf y)>\max(t,d)\}}d\mathbf{y} dt \nonumber\\
& = \int_0^\infty \int_{C^*} 1_{\{\psi^*(\mathbf y)>t\}}d\mathbf{y} dt \nonumber\\
& =  \int_{C^*} \psi^*(\mathbf{y})d\mathbf{y} \; . \label{eq:layer_cake}
\end{align}
Using this equality and our choice for $g$ we will rewrite the inequality from Lemma \ref{thm:riesz} as
\begin{equation}\label{eq:ineq}
\int_{C^*} \psi^*(\mathbf y)d\mathbf{y} \leq \int_{C^*} \bar{\psi}(\mathbf y) d\mathbf{y}
\end{equation}
where
$$\bar{\psi}(\mathbf y) = \int_{\mathbb{S}^{m-1}} f^*(\mathbf z)K\left(\langle \mathbf z/R,\mathbf y/R \rangle\right) d\mathbf{z} \; .$$

Note that both $\psi^*(\mathbf y)$ and $\bar{\psi}(\mathbf y)$ are spherically symmetric. More concretely, they both depend only on the angle $\angle(\mathbf y,\mathbf z_0)$, so in an abuse of notation we will write $\bar{\psi}(\alpha)$ and $\psi^*(\alpha)$ where $\alpha = \angle(\mathbf y,\mathbf z_0)$.

For convenience we will define a measure $\nu$ by
$$d\nu(\phi) = A_{m-2}(R\sin\phi)R d\phi$$
where $A_m(R)$ denotes the Haar measure of the $m$-sphere with radius $R$. We do this so that an integral like
$$\int_{\mathbb{S}^{m-1}} \psi^* d\mathbf{y} = \int_0^\pi \psi^*(\phi)A_{m-2}(R\sin\phi)Rd\phi$$ can be expressed as
$$\int_0^\pi \psi^* d\nu \; .$$

\subsection{Proof of Theorem \ref{thm:strongisoperimetryshell} (The Shell Case)}
Let $A\subseteq \mathbb{L}^m$ be a given subset with effective angle $\theta$. In order to apply Lemma \ref{thm:riesz}, note that
\begin{align*}
 |A \cap \text{ShellCap}(\mathbf y, \omega+\epsilon)| = \int_{\mathbb R^m} 1_{A \cap \text{ShellCap}(\mathbf y, \omega+\epsilon)}(\mathbf z) \; d\mathbf z \\
 =   \int_{\mathbb S^{m-1}} \left(\int_{R_L}^{R_U} \left(\frac{r}{R}\right)^{m-1}1_{A\cap \text{ShellCap}(\mathbf y, \omega+\epsilon)}\left(\frac{r}{R}\mathbf z\right) \; dr \right) d\mathbf z
\end{align*}
by using spherical coordinates, so that if we define
\begin{equation}\label{def:fA}
f_A(\mathbf z) = \int_{R_L}^{R_U} \left(\frac{r}{R}\right)^{m-1}1_A\left(\frac{r}{R}\mathbf z\right) \; dr
\end{equation}
for $A\subseteq \mathbb{L}^m$ and
$$K(\cos \alpha) = \begin{cases} 0 & \omega+\epsilon<\alpha\leq \pi \\ 1 & 0\leq \alpha \leq \omega+\epsilon \end{cases},$$
then
\begin{align*}
\psi(\mathbf y) & = |A \cap \text{ShellCap}(\mathbf y, \omega+\epsilon)| \\
 & = \int_{\mathbb S^{m-1}} f_A(\mathbf z)K(\langle \mathbf z/R, \mathbf y/R\rangle) d\mathbf z  \; .
\end{align*}
Both $\psi$ and $f_A$ can be thought of as functions on the sphere $\mathbb{S}^{m-1}$. Let $\psi^*,f_A^*$ be their respective symmetric decreasing rearrangements about a pole $\mathbf z_0$.
Define
\begin{align*}
\bar\psi(\mathbf y) & = \int_{\mathbb S^{m-1}} f^*_A(\mathbf z)K(\langle \mathbf z/R, \mathbf y/R\rangle) d\mathbf z
\end{align*}
so that by definition we have \eqref{eq:ineq}.

The inequality \eqref{eq:ineq} allows to compare $\psi$ and $\bar\psi$, but we require a way to compare $\psi$ with the function arising from a shell cap of angle $\theta$. Let
$$A' = \text{ShellCap}(\mathbf z_0,\theta)$$
and
$$\bar{\bar\psi}(\mathbf y)  = |A' \cap \text{ShellCap}(\mathbf y, \omega+\epsilon)|  \; .$$
We will show that 
\begin{equation} \label{eq:ineq2}
\int_{C^*} \bar{\psi}(\mathbf y)d\mathbf{y} \leq \int_{C^*} \bar{\bar\psi}(\mathbf y) d\mathbf{y}
\end{equation}
so that along with \eqref{eq:ineq},
\begin{equation} \label{eq:ineq3}
\int_{C^*} \psi^*(\mathbf y)d\mathbf{y} \leq \int_{C^*} \bar{\bar\psi}(\mathbf y) d\mathbf{y} \; .
\end{equation}
To show the inequality \eqref{eq:ineq2} note
\begin{align} \label{eq:construct}
\int_{C^*} & \bar{\psi}(\mathbf y)d\mathbf{y} \nonumber\\
& = \int_{\mathbb{S}^{m-1}}\int_{\mathbb{S}^{m-1}}1_{C^*}(\mathbf y)f_{A}^*(\mathbf z) K(\langle \mathbf z/R, \mathbf y/R\rangle) d\mathbf y d\mathbf z \nonumber \\
& = \int_{\mathbb{S}^{m-1}} f_{A}^*(\mathbf z) \left(  \int_{\mathbb{S}^{m-1}}1_{C^*}(\mathbf y)K(\langle \mathbf z/R, \mathbf y/R\rangle) d\mathbf y \right) d\mathbf z \; .
\end{align}
The term inside the parentheses is the measure of the intersection between the cap $C^*$ centered at $\mathbf z_0$ and a cap of angle $\omega+\epsilon$ centered at $\mathbf z$. This intersection measure is a function only of the angle $\angle(\mathbf z_0,\mathbf z)$ and is nonincreasing in that angle. Consider functions $f:\mathbb{S}^{m-1}\to\mathbb{R}$ with $0\leq f(\mathbf z)\leq \int_{R_L}^{R_U} \left(\frac{r}{R}\right)^{m-1}dr$ and $\int f(\mathbf z)d\mathbf z = |A|$. Both $f_{A}^*$ and $f_{A'}$ satisfy these properties and moreover $f_{A'}$ is extremal in the sense that $f_{A'}(\mathbf z)=\int_{R_L}^{R_U} \left(\frac{r}{R}\right)^{m-1}dr$ when $\angle(\mathbf z_0,\mathbf z)\leq\theta$ and $0$ when $\angle(\mathbf z_0,\mathbf z)>\theta$. Therefore \eqref{eq:construct} is maximized by replacing $f_{A}^*$ with $f_{A'}$, and 
\begin{align*}
\int_{C^*} & \bar{\psi}(\mathbf y)d\mathbf{y} \\
& = \int_{\mathbb{S}^{m-1}} f_{A}^*(\mathbf z) \left(  \int_{\mathbb{S}^{m-1}}1_{C^*}(\mathbf y)K(\langle \mathbf z/R, \mathbf y/R\rangle) d\mathbf y \right) d\mathbf z \nonumber\\
& \leq  \int_{\mathbb{S}^{m-1}} f_{A'}(\mathbf z) \left(  \int_{\mathbb{S}^{m-1}}1_{C^*}(\mathbf y)K(\langle \mathbf z/R, \mathbf y/R\rangle) d\mathbf y \right) d\mathbf z \nonumber \\
& = \int_{C^*} \bar{\bar\psi}(\mathbf y) d\mathbf{y} \; .
\end{align*}

Equipped with \eqref{eq:ineq3}, we are now ready to finish the proof of Theorem~\ref{thm:strongisoperimetryshell}. Proposition \ref{P:measurecon} implies that for any $0 < \epsilon < 1$, there exists an $M(\epsilon)$ such that for $m > M(\epsilon)$ we have
\begin{eqnarray} \label{eq:con1}
\mathbb P \left( \angle(\mathbf z_0,\mathbf Y) \in [\pi/2-\epsilon,\pi/2+\epsilon] \right) & \geq & 1-\frac{\epsilon^2}{2}
\end{eqnarray}
where $\mathbf Y$ is drawn from any rotationally invariant distribution on $\mathbb{L}^m$. Because the random quantity $|A\cap\text{ShellCap}(\mathbf Y,\omega+\epsilon)|$ depends only on the direction of $\mathbf Y$, and not on its magnitude, we can instead consider $\mathbf Y$ to be distributed according to the Haar measure on $\mathbb{S}^{m-1}$.
The constant $M(\epsilon)$ is determined only by the concentration of measure phenomenon cited above, and it does not depend on any parameters in the problem other than $\epsilon$. From now on, let us restrict our attention to dimensions $m > M(\epsilon)$. Due to the triangle inequality for the geodesic metric, for $\mathbf y$ such that $\angle(\mathbf z_0,\mathbf y) \in [\pi/2-\epsilon,\pi/2+\epsilon]$ we have
$$
A' \cap \text{ShellCap}(\mathbf y_0, \omega) \subseteq A' \cap \text{ShellCap}(\mathbf y,\omega+\epsilon)
$$
where $\mathbf y_0$ is such that $\angle(\mathbf z_0,\mathbf y_0)=\pi/2$. Therefore,
\begin{equation} \label{eq:greaterV}
 \bar{\bar{\psi}}(\angle(\mathbf z_0,\mathbf y)) = |A' \cap \text{ShellCap}(\mathbf y, \omega + \epsilon)| \geq V
\end{equation}
for all for $\mathbf y$ such that $\angle(\mathbf z_0,\mathbf y) \in [\pi/2-\epsilon,\pi/2+\epsilon]$ and

\begin{eqnarray}
\mathbb P\left( \bar{ \bar{\psi}}(\mathbf{Y})\geq V\right) & = &
\mathbb P\left( |A'\cap \text{ShellCap}(\mathbf Y,\omega+\epsilon)| \geq V \right) \nonumber \\
& \geq & 1-\frac{\epsilon^2}{2} \nonumber\\
& \geq & 1 - \frac{\epsilon}{2} \label{eq:con2} \; .
\end{eqnarray}

To prove the theorem, we need to show that
\begin{align} \label{eq:lem31goal1}
\mathbb P &\left(  \psi(\mathbf Y) >  (1-\epsilon)V \right) \nonumber\\
& = \mathbb P\left( |A\cap \text{ShellCap}(\mathbf Y,\omega+\epsilon)|  > (1-\epsilon)V \right)\nonumber \\
& \geq 1-\epsilon
\end{align}
for any arbitrary set $A\subseteq \mathbb{L}^m$. Recall that by the definition of a decreasing symmetric rearrangement, we have
\begin{equation*}
\mathbb P\left(  \psi^*(\mathbf Y)>d \right) = \mathbb P\left( \psi(\mathbf Y)>d \right)
\end{equation*}
for any threshold $d$ and this implies
\begin{equation} \label{eq:levelseteq}
\mathbb P\left(  \psi^*(\mathbf Y) \leq (1-\epsilon)V \right) = \mathbb P\left( \psi(\mathbf Y) \leq (1-\epsilon)V \right) \; .
\end{equation}
 Therefore, the desired statement in \eqref{eq:lem31goal1} can be equivalently written as
\begin{equation} \label{eq:lem31step1}
\mathbb P\left(  \psi^*(\mathbf Y) \leq (1-\epsilon)V \right) \leq \epsilon.
\end{equation}

{Turning to proving \eqref{eq:lem31step1}, recall that by the definition  of a decreasing symmetric rearrangement, $\psi^*(\alpha)$ is nonincreasing in the angle $\alpha=\angle(\mathbf y,\mathbf z_0)$ over the interval $0\leq \alpha \leq \pi$. Let $\beta$ be the smallest value such that $\psi^*(\beta) = (1-\epsilon)V$, or more explicitly, $$\beta = \inf\{\alpha : \psi^*(\alpha) \leq (1-\epsilon)V\} \; .$$ If  $\beta \geq \pi/2 + \epsilon$, then \eqref{eq:lem31step1} would follow trivially from \eqref{eq:con1} and the fact that $\psi^*(\alpha)$ would be greater than $(1-\epsilon)V$ for all $0<\alpha<\pi/2 + \epsilon$. We will therefore assume that $0< \beta < \pi/2 + \epsilon$. It remains to show that even if this is the case, we have \eqref{eq:lem31step1}.}

By the definition of $\beta$ and the fact that $\psi^*$ is nonincreasing,
\begin{align}
\mathbb P\left(  \psi^*(\mathbf Y) \leq (1-\epsilon)V \right) & =  \frac{1}{A_{m-1}(R)}\int^\pi_\beta d\nu \nonumber\\
& =  \frac{1}{A_{m-1}(R)}\int^{\max\{\beta,\frac{\pi}{2}-\epsilon\}}_\beta d\nu \nonumber \\
& + \frac{1}{A_{m-1}(R)}\int^{\frac{\pi}{2}+\epsilon}_{\max\{\beta,\frac{\pi}{2}-\epsilon\}} d\nu \nonumber \\
& +  \frac{1}{A_{m-1}(R)}\int^\pi_{\frac{\pi}{2}+\epsilon} d\nu \label{eq:psilessthanV2} \; .
\end{align}
To bound the first and third terms of \eqref{eq:psilessthanV2} note that
\begin{align}
 \frac{1}{A_{m-1}(R)}\int^{\max\{\beta,\frac{\pi}{2}-\epsilon\}}_\beta d\nu + \frac{1}{A_{m-1}(R)}\int^\pi_{\frac{\pi}{2}+\epsilon} d\nu & \leq  \frac{\epsilon^2}{2} \label{eq:conineq0}\\
& \leq  \frac{\epsilon}{2} \label{eq:conineq}
\end{align}
as a consequence of \eqref{eq:con1}. In order to bound the second term, we establish the following chain of (in)equalities which will be justified below.
\begin{align}
\frac{1}{A_{m-1}(R)}\int^\pi_{\frac{\pi}{2}+\epsilon}d\nu & \geq \frac{1}{(1-\epsilon)V A_{m-1}(R)}\int^\pi_{\frac{\pi}{2}+\epsilon}(\psi^*-\bar{\bar{\psi}})d\nu \label{eq:step1}\\
& =  \frac{1}{(1-\epsilon)V A_{m-1}(R)}\int_0^{\frac{\pi}{2}+\epsilon}(\bar{\bar{\psi}}-\psi^*)d\nu \label{eq:step2}\\
& \geq  \frac{1}{(1-\epsilon)V A_{m-1}(R)}\int_{\beta}^{\frac{\pi}{2}+\epsilon}(\bar{\bar{\psi}}-\psi^*)d\nu \label{eq:step3}\\
& \geq  \frac{\epsilon}{(1-\epsilon)A_{m-1}(R)}\int_{\max\{\beta,\frac{\pi}{2}-\epsilon\}}^{\frac{\pi}{2}+\epsilon}d\nu \label{eq:step4} \\
& \geq \frac{\epsilon}{A_{m-1}(R)}\int_{\max\{\beta,\frac{\pi}{2}-\epsilon\}}^{\frac{\pi}{2}+\epsilon}d\nu \label{eq:step5}
\end{align}
Combining \eqref{eq:step5} with \eqref{eq:conineq0} reveals that the second term in \eqref{eq:psilessthanV2} is also bounded by $\epsilon/2$, therefore
$$\mathbb P\left(  \psi^*(\mathbf Y) \leq (1-\epsilon)V \right)$$
must be bounded by $\epsilon$, which proves Theorem \ref{thm:strongisoperimetryshell}.

The first inequality \eqref{eq:step1} is a consequence of the fact that over the range of the integral, $\psi^*$ is less than or equal to $(1-\epsilon)V$ and $\bar{\bar{\psi}}$ is non-negative. The equality in \eqref{eq:step2} follows from
$$\int_0^\pi \psi^*d\nu = \int_0^\pi \bar{\bar{\psi}}d\nu \; ,$$
which is itself a consequence of \eqref{eq:layer_cake} with $C = \mathbb{S}^{m-1}$ and
\begin{align}
\int_{\mathbb{S}^{m-1}} \psi(\mathbf y) d\mathbf y & =\int_{\mathbb{S}^{m-1}} \int_{\mathbb{S}^{m-1}} f_A(\mathbf z)  K(\langle\mathbf z/R, \mathbf y/R\rangle)d\mathbf zd\mathbf y \nonumber\\
& = \int \int K(\langle \mathbf y/R,\mathbf z/R\rangle)d\mathbf y f_A(\mathbf z)d\mathbf z \nonumber\\
& =  \int \mu(\text{Cap}(\mathbf y,\omega)) f_A(\mathbf z) d\mathbf z \nonumber\\
& =  \mu(\text{Cap}(\mathbf y,\omega))|A| \nonumber\\
& =  \int \mu(\text{Cap}(\mathbf y,\omega)) f_{A'}(\mathbf z) d\mathbf z\nonumber \\
& =  \int \int f_{A'}(\mathbf z)K(\langle\mathbf z/R,\mathbf y/R \rangle) d\mathbf zd\mathbf y \nonumber\\ 
&= \int_{\mathbb{S}^{m-1}} \bar{\bar{\psi}}(\mathbf y) d\mathbf y \label{eq:totalmeas}
 \; .
\end{align}
Next we have \eqref{eq:step3} which is due to the rearrangement inequality \eqref{eq:ineq3} when $C$ is the super-level set $\{{\mathbf y} : \psi(\mathbf y) > (1-\epsilon)V\}$. By the definition of a symmetric decreasing rearrangement, $\mu(\{{\mathbf y} : \psi(\mathbf y) > (1-\epsilon)V\}) = \mu(\{{\mathbf y} : \psi^*(\mathbf y) > (1-\epsilon)V\})$, and the set on the right-hand side is an open or closed spherical cap of angle $\beta$. Thus $C^*$ is a spherical cap with angle $\beta$ and the rearrangement inequality \eqref{eq:ineq3} gives $$\int_0^\beta \psi^* d\nu \leq \int_0^\beta \bar{\bar\psi} d\nu \; .$$
Finally, for the inequality \eqref{eq:step4}, we first replace the lower integral limit with $\max\{\beta,\pi/2-\epsilon\} \geq \beta$. Then $\bar{\bar{\psi}} \geq V$ over the range of the integral due to \eqref{eq:greaterV}.  Additionally, $\psi^* \leq (1-\epsilon)V$ over the range of the integral, and the inequality follows.

\subsection{Proof of Theorem \ref{thm:strongisoperimetrysphere} (The Sphere Case)} 
Given any $A \subseteq \mathbb{S}^{m-1}$ with effective angle $\theta > 0$, construct a corresponding
$$A_\text{shell} = \left\{ {\bf y} \in \mathbb{L}^m : \; R\frac{{\bf y}}{\|\bf{y}\|} \in A \right\} \; .$$
The set $A_\text{shell}$ also has effective angle $\theta$ as a subset of $\mathbb{L}^m$ since
\begin{align*}
|A_\text{shell}| &= \int_{\mathbb R^m} 1_{A_\text{shell}} (\mathbf z) \; d\mathbf z \\
& =   \int_{\mathbb S^{m-1}} \left(\int_{R_L}^{R_U} \left(\frac{r}{R}\right)^{m-1}1_{A_\text{shell}}\left(\frac{r\mathbf z}{R}\right) \; dr \right) d\mathbf z \\
& =   \int_{\mathbb S^{m-1}} 1_{A}(\mathbf z)d\mathbf z \int_{R_L}^{R_U} \left(\frac{r}{R}\right)^{m-1}dr \\
& = \mu(A) \int_{R_L}^{R_U} \left(\frac{r}{R}\right)^{m-1}dr \\
& = \mu(\text{Cap}(\mathbf y,\theta)) \int_{R_L}^{R_U} \left(\frac{r}{R}\right)^{m-1}dr \\
& = \int_{\mathbb R^m} 1_{\text{ShellCap}(\mathbf y,\theta)} (\mathbf z) \; d\mathbf z \\
& = |\text{ShellCap}(\mathbf y,\theta)| \; .
\end{align*}
For any $\epsilon >0$, we can apply Theorem \ref{thm:strongisoperimetryshell} to find an $M(\epsilon)$ such that for $m>M(\epsilon)$,
\begin{equation} \label{eq:cor21_1}
\mathbb{P}\left(|A_\text{shell}\cap \text{ShellCap}(\mathbf Y,\omega+\epsilon)|> (1-\epsilon)V_\text{shell}\right)\geq 1-\epsilon \; ,
\end{equation}
where $V_\text{shell} =|\text{ShellCap}(\mathbf z_0, \theta) \cap \text{ShellCap}(\mathbf y_0, \omega)|$ with $\angle(\mathbf z_0,\mathbf y_0)=\pi/2$. Because the set $\text{ShellCap}(\mathbf y, \omega)$ depends only on the direction of $\mathbf y$, and not on its magnitude, the probability in \eqref{eq:cor21_1} is the same whether we consider $\mathbf Y$ to be uniformly distributed on $\mathbb{S}^{m-1}$ or from some rotationally invariant probability distribution on $\mathbb{L}^m$. Using spherical coordinates, we have
\begin{align*}
 & |A_\text{shell} \cap \text{ShellCap}(\mathbf y, \omega+\epsilon)| \\
& = \int_{\mathbb R^m} 1_{A_\text{shell} \cap \text{ShellCap}(\mathbf y, \omega+\epsilon)}(\mathbf z) \; d\mathbf z \\
& =   \int_{\mathbb S^{m-1}} \left(\int_{R_L}^{R_U} \left(\frac{r}{R}\right)^{m-1}1_{A_\text{shell}\cap \text{ShellCap}(\mathbf y, \omega+\epsilon)}\left(\frac{r\mathbf z}{R}\right) \; dr \right) d\mathbf z \\
& =   \int_{\mathbb S^{m-1}} 1_{A \cap \text{Cap}(\mathbf y, \omega+\epsilon)}(\mathbf z)d\mathbf z \int_{R_L}^{R_U} \left(\frac{r}{R}\right)^{m-1}dr \\
& =   \mu(A \cap \text{Cap}(\mathbf y, \omega+\epsilon)) \int_{R_L}^{R_U} \left(\frac{r}{R}\right)^{m-1}dr
\end{align*}
and similarly,
\begin{align*}
& |\text{ShellCap}(\mathbf z_0, \theta) \cap \text{ShellCap}(\mathbf y_0, \omega)| \\
& =\int_{\mathbb S^{m-1}} 1_{\text{Cap}(\mathbf z_0, \theta) \cap \text{Cap}(\mathbf y_0, \omega)}(\mathbf z) d\mathbf z  \int_{R_L}^{R_U} \left(\frac{r}{R}\right)^{m-1}dr  \\
& = \mu(\text{Cap}(\mathbf z_0, \theta) \cap \text{Cap}(\mathbf y_0, \omega)) \int_{R_L}^{R_U} \left(\frac{r}{R}\right)^{m-1}dr\; .
\end{align*}
By dividing out the $\int_{R_L}^{R_U} \left(\frac{r}{R}\right)^{m-1}dr$ term, \eqref{eq:cor21_1} implies
\begin{equation}
\mathbb{P}\left(\mu(A \cap \text{Cap}(\mathbf Y,\omega+\epsilon))> (1-\epsilon)V\right)\geq 1-\epsilon
\end{equation}
where $V=\mu(\text{Cap}(\mathbf z_0, \theta) \cap \text{Cap}(\mathbf y_0, \omega))$ as desired.

\section{Proofs of Typicality Lemmas}\label{A:Proof_Typicality}
Here we prove the typicality lemmas presented in Section \ref{sec:typicality}.

\subsection{Proof of Lemma \ref{L:E1}}\label{SS:typicalityproof1}
Recalling that $\mathbf Z=[Z^n(1),Z^n(2),\ldots,Z^n(B)]$, we have
$$\|\mathbf Z \|^2 =\sum_{b=1}^{B} \|Z^n(b)\|^2.$$
Therefore by the weak law of large numbers, for any $\delta>0$ and $B$ sufficiently large we have
\begin{align*}
\text{Pr}\left(  \left|\frac{1}{B} \|\mathbf Z \|^2 - E[  \|Z^n \|^2] \right|\leq \delta  \right)\geq 1-\delta,
\end{align*}
i.e.,
\begin{align*}
\text{Pr}( \| \mathbf Z\|^2 \in [nB(P+N- \delta), nB(P+N+ \delta)] )\geq 1- \delta,
\end{align*}
since by assumption $E[  \|X^n \|^2]=nP$ and thus $E[\|Z^n\|^2]=n(P+N)$. Because $\mathbf Z$ and $\mathbf Y$ are identically distributed, the above relation also holds with $\| \mathbf Z\|^2$ replaced by $\| \mathbf Y\|^2$. This completes the proof of the lemma.

\subsection{Proof of Lemma \ref{L:prob_s(x,i)}}\label{SS:ProofStrongTyp}

We now present the proof of Lemma \ref{L:prob_s(x,i)}. By the law of large numbers and Lemma \ref{L:E1}, we have
for any $\epsilon >0$ and sufficiently large $B$,
\begin{align*}
\mbox{Pr}((\mathbf{X},\mathbf{Z})\in  S_{\epsilon}(X^n,Z^n))\geq 1-\epsilon
\end{align*}
where
\begin{align*}
  &S_{\epsilon}(X^n,Z^n)\\
 &:=\Big \{(\mathbf{x},\mathbf{z} ):   
  \|\mathbf x-\mathbf z\| \in [\sqrt{nB(N-\epsilon)} ,  \sqrt{nB(N+\epsilon)}  ],\\
 &~~~~~~~ \mathbf z \in \text{Ball}\left( \mathbf 0, \sqrt{nB(P+N+\epsilon)}  \right),\\
 &~~~~~~~ 2^{ nB(  \log \text{sin} \theta_n-\epsilon)}\leq p( f(\mathbf{z})|\mathbf{x})\leq 2^{ nB( \log \text{sin} \theta_n+\epsilon)}
 \Big \}.
\end{align*}

Note that in terms of $S_{\epsilon}(X^n,Z^n)$, the set $S_{\epsilon}(Z^n|\mathbf{x},\mathbf{i})$ in Lemma \ref{L:prob_s(x,i)} can be simply written as
 $$S_{\epsilon}(Z^n|\mathbf{x},\mathbf{i})=\{ \mathbf{z}:  f(\mathbf{z})=\mathbf{i}, (\mathbf{x},\mathbf{z})\in S_{\epsilon}(X^n,Z^n)   \}.$$
Therefore, for $B$ sufficiently large, we have
\begin{align*}
 \mbox{Pr}( \mathbf{Z} \notin S_{\epsilon}(Z^n|\mathbf{X},\mathbf{I})   )&=\mbox{Pr}( f(\mathbf{Z})=\mathbf{I},  (\mathbf{X},\mathbf{Z})\notin  S(X^n,Z^n)    )\\
&\leq  {\epsilon}.
\end{align*}
On the other hand, defining $S_{\epsilon}(X^n,I_n):= \{(\mathbf{x},\mathbf{i}): \mbox{Pr}(\mathbf{Z} \in S_{\epsilon}(Z^n|\mathbf{x},\mathbf{i}) | \mathbf{x},\mathbf{i}   )\geq 1-\sqrt{\epsilon}  \}$, we have
\begin{align*}
& \mbox{Pr}( \mathbf{Z} \notin S_{\epsilon}(Z^n|\mathbf{X},\mathbf{I})    )\\
&=\sum_{(\mathbf{x},\mathbf{i})\in S_{\epsilon}(X^n,I_n)}\mbox{Pr}(\mathbf{Z} \notin S_{\epsilon}(Z^n|\mathbf{x},\mathbf{i}) | \mathbf{x},\mathbf{i}   ) p(\mathbf{x},\mathbf{i})\\
&~~~~+\sum_{(\mathbf{x},\mathbf{i})\notin S_{\epsilon}(X^n,I_n)}\mbox{Pr}(\mathbf{Z} \notin S_{\epsilon}(Z^n|\mathbf{x},\mathbf{i}) | \mathbf{x},\mathbf{i}   ) p(\mathbf{x},\mathbf{i})\\
&\geq \sqrt{\epsilon}\cdot \mbox{Pr}( S_{\epsilon}^c (X^n,I_n)  ).
\end{align*}
Therefore, we have for $B$ sufficiently large,
$$\mbox{Pr}( S_{\epsilon}^c (X^n,I_n)  ) \leq  \frac{\epsilon}{\sqrt{\epsilon}}=  \sqrt{\epsilon},$$
and thus $$\mbox{Pr}(S_{\epsilon}(X^n,I_n))\geq 1- \sqrt{\epsilon},$$
which proves \dref{E:part1}.

To prove \dref{E:part2}, consider any $(\mathbf{x},\mathbf{i}) \in S_{\epsilon}(X^n,I_n)$. From the definition of $S_{\epsilon}(X^n,I_n)$,
$\mbox{Pr}(S_{\epsilon}(Z^n|\mathbf{x},\mathbf{i})|\mathbf{x},\mathbf{i})\geq 1-\sqrt{\epsilon}$. Therefore,
$S_{\epsilon}(Z^n|\mathbf{x},\mathbf{i})$ must be nonempty, i.e., there exists at least one $\mathbf{z}\in S_{\epsilon}(Z^n|\mathbf{x},\mathbf{i})$.
Consider any $\mathbf{z}\in S_{\epsilon}(Z^n|\mathbf{x},\mathbf{i})$. By the definition of $S_{\epsilon}(Z^n|\mathbf{x},\mathbf{i})$,
 we have $f(\mathbf{z})=\mathbf{i}$ and $(\mathbf{x},\mathbf{z})\in S_{\epsilon}(X^n,Z^n)$.  Then, it follows from the definition of $S_{\epsilon}(X^n,Z^n)$ that
$$2^{ nB(  \log \text{sin} \theta_n-\epsilon)}\leq p( f(\mathbf{z})|\mathbf{x})=p(  \mathbf{i} |\mathbf{x})\leq 2^{ nB( \log \text{sin} \theta_n+\epsilon)}.$$
This further implies that
\begin{align*}
 &\mbox{Pr}(\mathbf{Z}\in S_{\epsilon}(Z^n|\mathbf{x},\mathbf{i})|\mathbf{x})\\
=\ &\frac{\mbox{Pr}(f(\mathbf{Z})=\mathbf{i}|\mathbf{x})\mbox{Pr}(\mathbf{Z}\in S_{\epsilon}(Z^n|\mathbf{x},\mathbf{i})|\mathbf{x},f(\mathbf{Z})=\mathbf{i})}
             {\mbox{Pr}(f(\mathbf{Z})=\mathbf{i}|\mathbf{Z}\in S_{\epsilon}(Z^n|\mathbf{x},\mathbf{i}),\mathbf{x})}\\
=\ &p( \mathbf{i}|\mathbf{x}) \mbox{Pr}(S_{\epsilon}(Z^n|\mathbf{x},\mathbf{i})|\mathbf{x}, \mathbf{i})\\
\geq \ & 2^{ nB(  \log \text{sin} \theta_n-\epsilon)} (1-\sqrt{\epsilon})\\
\geq\ &2^{ nB(  \log \text{sin} \theta_n-2\epsilon)}
\end{align*}
for sufficiently large $B$, which concludes the proof of \dref{E:part2} and Lemma \ref{L:prob_s(x,i)}.

\subsection{Proof of Corollary \ref{C:isoperi}}\label{SS:ProofCor}
Let the effective angle of $A$ be denoted by $\theta'$, i.e.,
$$|A|=|\text{ShellCap}(\mathbf z_0,\theta')|$$ for some  $$\mathbf z_0\in  \mbox{Shell}\left(\mathbf 0,\sqrt{m(N-\epsilon)}  ,  \sqrt{m(N+\epsilon)}  \right),$$  where 
\begin{align*}
&\text{ShellCap}(\mathbf z_0,\theta')\\
&:=\bigg\{\mathbf z\in \mbox{Shell}(\mathbf 0,\sqrt{m(N-\epsilon)}  ,  \sqrt{m(N+\epsilon)}  ):\\
&~~~~~~~~~~~~~~~~~  \angle(\mathbf z_0,\mathbf z)\leq \theta' \bigg\}.
\end{align*}
Then using the formula for the volume of a shell cap (c.f. Appendix \ref{AS:singlecap} and in particular \dref{AE:shellcapvol}), we have
\begin{align*}
 |A| \leq 2^{\frac{m}{2}[\log (2\pi e (N+\epsilon)\text{sin}^2  \theta' )+\epsilon_1 ]}
\end{align*}
for some $\epsilon_1 \to 0$ as $m\to \infty$. Recall that by assumption 
\begin{align*}
|A| \geq  2^{\frac{m}{2}[\log (2\pi e (N+\epsilon)\text{sin}^2  \theta )  ]} ,
\end{align*}
and we hence have
\begin{align*}
\theta'\geq \theta - \epsilon_2
\end{align*}
for some $\epsilon_2 \to 0$ as $m\to \infty$.

We now apply Theorem~\ref{thm:strongisoperimetryshell} to this specific shell and subset $A$. First, using the formula of the intersection volume of two shell caps (c.f. Appendices \ref{AS:twocap} and in particular Lemma \ref{L:volumeintersectionshellcap}), we have 
\begin{align*}
&|\text{ShellCap}(\mathbf z_0, \theta') \cap \text{ShellCap}(\mathbf y_0, \omega)|\\
&\geq 2^{\frac{m}{2}[\log(2\pi e N (\text{sin}^2 \theta'  - \cos^2 \omega))-\epsilon_3]} \\
&\geq 2^{\frac{m}{2}[\log(2\pi e N  (\text{sin}^2 \theta   - \cos^2 \omega))-\epsilon_4]}
\end{align*}
for some $\epsilon_3, \epsilon_4 \to 0$ as $m\to \infty$, where $\angle(\mathbf z_0,\mathbf y_0)=\pi/2$
and $\theta'+\omega>\pi/2$. Then Theorem~\ref{thm:strongisoperimetryshell}  asserts that for any $\omega\in (\pi/2-\theta',\pi/2]$ and $m$ sufficiently large,
\begin{align*}
&\text{Pr}\Big(|A\cap \text{ShellCap}(\mathbf Y,\omega+\epsilon)| \\
&~~~~~\geq (1-\epsilon)2^{\frac{m}{2}[\log(2\pi eN(\text{sin}^2 \theta   - \cos^2 \omega))-\epsilon_4]}\Big)\geq 1-\epsilon,
\end{align*}
where $\mathbf Y$ is a random vector drawn from any rotationally invariant distribution on the shell.  Since 
$\pi/2-\theta'\leq \pi/2-\theta+\epsilon_2$, the condition $\omega\in (\pi/2-\theta',\pi/2]$ in the above can be replaced with the weaker condition $\omega\in (\pi/2-\theta+\epsilon_2,\pi/2]$. Now by choosing $m$ sufficiently large we can make $\epsilon_2,\epsilon_4$ and $\frac{2}{m}\log(1-\epsilon)$ as small as desired, so we have
\begin{align*}
&\text{Pr}\left(|A\cap \text{ShellCap}(\mathbf Y,\omega+\epsilon)|\geq 2^{\frac{m}{2}[\log(2\pi eN(\text{sin}^2 \theta   - \cos^2 \omega))-\epsilon]}\right)\\
&\geq 1-\epsilon,
\end{align*}
for any $\omega\in (\pi/2-\theta,\pi/2]$ and $m$ sufficiently large.  Finally, observe that for any $\mathbf y$ in the considered shell,
\begin{align*} 
&\text{ShellCap}(\mathbf y, \omega+\epsilon)\\
&\subseteq \text{Ball}\left(\mathbf y,  2\sqrt{m(N+\epsilon)} \sin \frac{\omega+\epsilon}{2} +2\sqrt{m\epsilon}\right).
\end{align*} 
This simply follows from the geometry illustrated in  Fig.~\ref{F:capinball} combined with the triangle inequality and the fact that the thickness of the shell can be trivially bounded by $2\sqrt{m\epsilon}$. Therefore, we can conclude that 
\begin{align*}
\text{Pr}\Bigg(& \left|A\cap \text{Ball}\left(\mathbf Y,  2\sqrt{m(N+\epsilon)} \sin \frac{\omega+\epsilon}{2} +2\sqrt{m\epsilon} \right)\right|\\
&\geq 2^{\frac{m}{2}[\log(2\pi eN(\text{sin}^2 \theta   - \cos^2 \omega))-\epsilon]}\Bigg)\geq 1-\epsilon
\end{align*}
for any $\omega\in (\pi/2-\theta,\pi/2]$ and $m$ sufficiently large. This completes the proof of  Corollary \ref{C:isoperi}.

\subsection{Proof of Lemma \ref{L:Keylemma}}\label{SS:appendixproofkeylemma}

Fix $\epsilon>0$ and consider a pair $(\mathbf x, \mathbf i)\in S_{\epsilon}(X^n,I_n)$.  From Lemma \ref{L:prob_s(x,i)}, we have 
$$\mbox{Pr}(\mathbf Z \in S_{\epsilon}(Z^n| \mathbf{x},\mathbf{i})|\mathbf{x})\geq 2^{ nB(  \log \text{sin} \theta_n-2\epsilon)},$$
for $B$ sufficiently large. We also have
\begin{align*}
\mbox{Pr}(\mathbf Z \in S_{\epsilon}(Z^n| \mathbf{x},\mathbf{i})|\mathbf{x})&\leq | S_{\epsilon}(Z^n| \mathbf{x},\mathbf{i})|\sup_{\mathbf z \in  S_{\epsilon}(Z^n| \mathbf{x},\mathbf{i})}p(\mathbf z|\mathbf x)\\
&\leq |  S_{\epsilon}(Z^n| \mathbf{x},\mathbf{i})| 2^{-nB\left(\frac{1}{2}\log 2\pi eN -\epsilon_1\right)},
\end{align*}
for some $\epsilon_1\to 0$ as $\epsilon\to 0$, where $p(\mathbf z|\mathbf x)$ refers to the conditional density of $\mathbf z$ given $\mathbf x$. The second inequality in the above follows because for any $\mathbf z \in  S_{\epsilon}(Z^n| \mathbf{x},\mathbf{i})$, we have
$$ { \|\mathbf x-\mathbf z\| \in [\sqrt{nB(N-\epsilon)} ,  \sqrt{nB(N+\epsilon)}  ]},$$
and therefore using the fact that $\mathbf Z$ is Gaussian distributed given $\mathbf x$, we have for any $\mathbf z \in  S_{\epsilon}(Z^n| \mathbf{x},\mathbf{i})$,
\begin{align*}
p(\mathbf z| \mathbf x)&=  \frac{1}{ (2\pi N)^{ \frac{nB}{2} } } e^{ -\frac{|| \mathbf z- \mathbf x  ||^2     }{2N}  } \\
&\leq    2^{ -\frac{     nB (N-\epsilon)     }{2N}  \log e  -\frac{nB}{2}\log 2\pi N   } \\
&=   2^{-nB\left(\frac{1}{2}\log 2\pi eN -\epsilon_1\right)}
\end{align*}
where $\epsilon_1\to 0$ as $\epsilon\to 0$. Therefore, for $B$ sufficiently large, the volume of $S_{\epsilon}(Z^n| \mathbf{x},\mathbf{i})$ can be lower bounded by
$$| S_{\epsilon}(Z^n| \mathbf{x},\mathbf{i})|\geq 2^{nB\left(\frac{1}{2}\log(2\pi eN\text{sin} ^2 \theta_n )-2\epsilon-\epsilon_1\right)}.$$

 Let $\theta'_n$ be defined such that
$$\log 2\pi e (N+\epsilon)  \text{sin}^2 \theta'_n=\frac{1}{2}\log(2\pi eN\text{sin} ^2 \theta_n )-2\epsilon-\epsilon_1.$$
Obviously, we have $\theta'_n\leq \theta_n$ and $\theta'_n\to \theta_n$ as $\epsilon\to 0$.
Noting that $ S_{\epsilon}(Z^n| \mathbf{x},\mathbf{i})$ is a subset of $$ \mbox{Shell}\left(\mathbf x,\sqrt{nB(N-\epsilon)} ,  \sqrt{nB(N+\epsilon)} \right),$$
by Corollary \ref{C:isoperi}, for any $\omega\in (\pi/2-\theta'_n, \pi/2]$ we have 
\begin{align}
\mbox{Pr}\Bigg( & \left|S_{\epsilon}(Z^n| \mathbf{x},\mathbf{i})\cap \text{Ball}\left(\mathbf U, \sqrt{nBN\left(4\text{sin}^2 \frac{\omega}{2}+\epsilon_2\right)}\right) \right| \nonumber \\
& \geq 2^{nB\left[\frac{1}{2}\log(2\pi eN(\text{sin} ^2 \theta_n - \cos^2 \omega))-\epsilon_3\right]}   \Bigg)\geq 1-\epsilon \label{E: transprob}
\end{align}
for any $\mathbf U$ drawn from a rotationally invariant distribution around $\mathbf x$ on $ \mbox{Shell}\left(\mathbf x,\sqrt{nB(N-\epsilon)} ,  \sqrt{nB(N+\epsilon)} \right)$, where $\epsilon_2$ is defined such that
$$\sqrt{nBN\left(4\text{sin}^2 \frac{\omega}{2}+\epsilon_2\right)}=2\sqrt{nB(N+\epsilon)} \sin \frac{\omega+\epsilon}{2}+2\sqrt{m\epsilon},$$
and $\epsilon_3$ is defined such that
\begin{align*}
&\frac{1}{2}\log(2\pi eN(\text{sin} ^2 \theta_n - \cos^2 \omega))-\epsilon_3\\
=\ &  \frac{1}{2}\log(2\pi eN(\text{sin} ^2 \theta'_n - \cos^2 \omega))-\epsilon,
\end{align*}
and both $\epsilon_2$ and $\epsilon_3$ tend to zero as $\epsilon$ goes to zero.

We now translate the bound \dref{E: transprob} on the probability involving a rotationally invariantly distributed $\mathbf U$ on the shell to a bound on the probability involving $\mathbf Y$. Define $\mathcal Y_{(\mathbf x,\mathbf i)}$ to be the following set of $\mathbf y$: 
\begin{align*}
  \Bigg\{ \mathbf y:
&  \left|S_{\epsilon}(Z^n| \mathbf{x},\mathbf{i})\cap \text{Ball}\left(\mathbf y, \sqrt{nBN\left(4\text{sin}^2 \frac{\omega}{2}+\epsilon_2\right)}\right) \right|   \\
&  \geq 2^{nB\left[\frac{1}{2}\log(2\pi eN(\text{sin} ^2 \theta_n - \cos^2 \omega))-\epsilon_3\right]} \Bigg\}  .
\end{align*}
Then we have for $(\mathbf x, \mathbf i)\in S_{\epsilon}(X^n,I_n)$ and $B$ sufficiently large,
\begin{align*}
&\mbox{Pr}(\mathbf Y\in \mathcal Y_{(\mathbf x,\mathbf i)}|\mathbf x)\\
&\geq \mbox{Pr}\Big(\mathbf Y\in \mathcal Y_{(\mathbf x,\mathbf i)}, \\
&~~~~~~~~~~~~~\mathbf Y\in \mbox{Shell}\left(\mathbf x,\sqrt{nB(N-\epsilon)} ,  \sqrt{nB(N+\epsilon)} \right)\Big|\mathbf x\Big)\\
&=\mbox{Pr}\left(\mathbf Y\in \mbox{Shell}\left(\mathbf x,\sqrt{nB(N-\epsilon)} ,  \sqrt{nB(N+\epsilon)} \right)\Big|\mathbf x\right)\\
&~~~\times \mbox{Pr}\Big(\mathbf Y\in \mathcal Y_{(\mathbf x,\mathbf i)} \Big|\mathbf x, \\
&~~~~~~~~~~~~~~\mathbf Y\in \mbox{Shell}\left(\mathbf x,\sqrt{nB(N-\epsilon)} ,  \sqrt{nB(N+\epsilon)} \right)\Big)\\
&\geq (1-\epsilon) \mbox{Pr}\Big(\mathbf Y\in \mathcal Y_{(\mathbf x,\mathbf i)} \Big|\mathbf x, \\
&~~~~~~~~~~~~~~~~~\mathbf Y\in \mbox{Shell}\left(\mathbf x,\sqrt{nB(N-\epsilon)} ,  \sqrt{nB(N+\epsilon)} \right)\Big)\\
&\geq (1-\epsilon)^2,
\end{align*}
where the second inequality simply follows by applying the law of large numbers in a manner similar to the proof of Lemma  \ref{L:E1}, and the  last inequality follows from combining  \dref{E: transprob}  and the fact that if $\mathbf x$ is known and $\mathbf Y$ is restricted to  $\mbox{Shell}\left(\mathbf x,\sqrt{nB(N-\epsilon)} ,  \sqrt{nB(N+\epsilon)} \right)$ then $\mathbf Y$ is rotationally invariant around $\mathbf x$  on this shell.

Since by definition $S_{\epsilon}(Z^n| \mathbf{x},\mathbf{i})$ is a subset of
$ f^{-1}(\mathbf i) \cap \text{Ball}\left( \mathbf 0, \sqrt{nB(P+N+\epsilon)}  \right),$
we  have
\begin{align*}
& \Bigg|f^{-1}(\mathbf i)   \cap \text{Ball}\left( \mathbf 0, \sqrt{nB(P+N+\epsilon)}  \right)\\
 &~~~~~\cap \text{Ball}\left(\mathbf y, \sqrt{nBN\left(4\text{sin}^2 \frac{\omega}{2}+\epsilon_2\right)}\right)   \Bigg| \\
 &\geq 2^{nB\left[\frac{1}{2}\log(2\pi eN(\text{sin} ^2 \theta_n - \cos^2 \omega))-\epsilon_3\right]}
\end{align*}
for any $\mathbf y\in \mathcal Y_{(\mathbf x,\mathbf i)} $, and therefore  for $B$ sufficiently large,
\begin{align*}
&\mbox{Pr}\Bigg(\Bigg| f^{-1}(\mathbf I) \cap \text{Ball}\left( \mathbf 0, \sqrt{nB(P+N+\epsilon)}  \right)\\
&~~~~~~~~~~~\cap \text{Ball}\left(\mathbf Y, \sqrt{nBN\left(4\text{sin}^2 \frac{\omega}{2}+\epsilon_2\right)}\right)   \Bigg| \\
&~~~~~~\geq 2^{nB\left[\frac{1}{2}\log(2\pi eN(\text{sin} ^2 \theta_n - \cos^2 \omega))-\epsilon_3\right]}\Bigg)\\
&\geq \sum_{(\mathbf x, \mathbf i)}\mbox{Pr}(\mathbf Y\in  \mathcal Y_{(\mathbf x,\mathbf i)} |\mathbf x)p(\mathbf x, \mathbf i)\\
&\geq \sum_{(\mathbf x, \mathbf i)\in S_{\epsilon}(X^n,I_n)}\mbox{Pr}(\mathbf Y\in \mathcal Y_{(\mathbf x,\mathbf i)} |\mathbf x)p(\mathbf x, \mathbf i)\\
&\geq (1-\epsilon)^2(1-\sqrt{\epsilon})\\
&\geq 1-4\sqrt{\epsilon},
\end{align*}
for any $\omega\in (\pi/2-\theta'_n, \pi/2]$.
Finally, choosing $\delta=\max\{ 4\sqrt{\epsilon},\epsilon_2, \epsilon_3, \theta_n-\theta'_n \}$ concludes the proof of Lemma \ref{L:Keylemma}. {Note that by choosing $B$ sufficiently large, $\epsilon$ and therefore $\delta$ can be made arbitrarily small.}

\section{Miscellaneous Results in High-Dimensional Geometry}\label{A:Miscellaneous}

This appendix derives some miscellaneous results in high-dimensional geometry, including the surface area (volume) of a spherical (shell) cap, the surface area (volume) of the intersection of two spherical (shell) caps, and the volume of the intersection of two balls.

\subsection{Surface Area (Volume) of A Spherical (Shell) Cap}\label{AS:singlecap}
We first derive the surface area (volume) formula for a spherical (shell) cap. See also \cite{singlecap}.

Let $C \subseteq \mathbb{S}^{m-1}$ be a spherical cap with angle $\theta$ on the $(m-1)$-sphere of radius $R=\sqrt{mN}$. The area $\mu(C)$ of $C$ can be written as
$$\mu(C) = \int_0^{\theta} A_{m-2}(R\sin\rho)Rd\rho$$
where $A_{m-2}(R\sin\rho)$ is the total surface area of the $(m-2)$-sphere of radius $R\sin\rho$. Plugging in the expression for the surface area of an $(m-2)$-sphere leads to
\begin{eqnarray*}
\mu(C) = \frac{2\pi^{\frac{m-1}{2}}}{\Gamma\left(\frac{m-1}{2}\right)}(mN)^{\frac{m-2}{2}}\int_0^{\theta}\sin^{m-2}\rho \; d\rho  .
\end{eqnarray*}
We now characterize the exponent of $\mu(C)$. First, by Stirling's approximation, $\frac{2\pi^{\frac{m-1}{2}}}{\Gamma\left(\frac{m-1}{2}\right)}(mN)^{\frac{m-2}{2}}$ in the above can be bounded as
\begin{align}
2^{\frac{m}{2}[\log (2\pi e N)-\epsilon_1]}\leq \frac{2\pi^{\frac{m-1}{2}}}{\Gamma\left(\frac{m-1}{2}\right)}(mN)^{\frac{m-2}{2}} \leq 2^{\frac{m}{2}[\log (2\pi e N)+\epsilon_1]}\label{E:citesphere}
\end{align}
for some $\epsilon_1 \to 0$ as $m\to \infty$. Also for $m$ sufficiently large, we have
\begin{align*}
\int_0^{\theta}\sin^{m-2}\rho  d\rho & =\int_0^{\theta}  2^{\frac{m-2}{2}\log \text{sin}^2 \rho}  d\rho \\
&\geq \int_{\theta-\frac{1}{m}}^{\theta}  2^{\frac{m-2}{2}\log \text{sin}^2 \rho}  d\rho \\
&\geq \frac{1}{m}  2^{\frac{m-2}{2}\log \text{sin}^2 (\theta-\frac{1}{m})}   \\
&\geq   2^{\frac{m}{2}\left(\log \text{sin}^2  \theta -\epsilon_2\right)}   \\
\end{align*}
and
\begin{align*}
\int_0^{\theta}\sin^{m-2}\rho  d\rho & =\int_0^{\theta}  2^{\frac{m-2}{2}\log \text{sin}^2 \rho}  d\rho \\
&\leq \theta\cdot 2^{\frac{m-2}{2}\log \text{sin}^2 \theta}   \\
&\leq   2^{\frac{m}{2}\left(\log \text{sin}^2  \theta +\epsilon_2\right)}   \\
\end{align*}
for some $\epsilon_2\to 0$ as $m\to \infty$. Therefore, the area $\mu(C)$ can be bounded as
\begin{align} \label{eq:appC_1}
2^{\frac{m}{2}[\log (2\pi e N\text{sin}^2  \theta )-\epsilon ]}\leq \mu(C) \leq 2^{\frac{m}{2}[\log (2\pi e N\text{sin}^2  \theta )+\epsilon ]}
\end{align}
for some $\epsilon\to 0$ as $m\to \infty$.

Now suppose that $C = \text{ShellCap}(\mathbf z_0,\theta) $ is a shell cap on
$$\mbox{Shell}\left(\mathbf 0,\sqrt{m(N-\delta)}  ,  \sqrt{m(N+\delta)}  \right)$$
where $\| \mathbf z_0 \|=  \sqrt{m(N-\delta)}$.
Let $R_L = \sqrt{m(N-\delta)}$, $R_U = \sqrt{m(N+\delta)}$ and define $\mathbb{S}_{R_L}^{m-1}$ to be the $m-1$ sphere of radius $R_L$ with Haar measure $\mu_{R_L}$. We use spherical coordinates to integrate over the surface areas of the individual caps that make up the shell cap,
\begin{align} \label{eq:appC_2}
 &|C| = \int_{\mathbb R^m} 1_{\text{ShellCap}(\mathbf z_0,\theta)} \; d\mathbf z \nonumber\\
& =   \int_{\mathbb S^{m-1}_{R_L}} \left(\int_{R_L}^{R_U} \left(\frac{r}{R_L}\right)^{m-1}1_{\text{ShellCap}(\mathbf z_0,\theta)}\left(\frac{r}{R_L}\mathbf z\right) \; dr \right) d\mathbf z \nonumber\\
& =   \int_{\mathbb S^{m-1}_{R_L}} 1_{\text{Cap}(\mathbf z_0,\theta)}(\mathbf z)d\mathbf z \int_{R_L}^{R_U} \left(\frac{r}{R_L}\right)^{m-1}dr \nonumber\\
& =   \mu_{R_L}(\text{Cap}(\mathbf z_0,\theta)) \int_{R_L}^{R_U} \left(\frac{r}{R_L}\right)^{m-1}dr \;
\end{align}
where the integral term on the right is bounded as
\begin{align} \label{eq:appC_3}
\int_{R_L}^{R_U} \left(\frac{r}{R_L}\right)^{m-1}dr & \geq \left(\sqrt{m(N+\delta)}-\sqrt{m(N-\delta)}\right) \; .
\end{align}

Together with \eqref{eq:appC_2}, \eqref{eq:appC_1} and \eqref{eq:appC_3} imply
$$|C| \geq 2^{\frac{m}{2}[\log(2\pi e(N-\delta)\text{sin}^2 \theta)-\epsilon]}$$
for sufficiently large $m$. In a similar way,
$$|C| \leq 2^{\frac{m}{2}[\log(2\pi e(N+\delta)\text{sin}^2 \theta)+\epsilon]} \; ,$$
and therefore
\begin{align}
2^{\frac{m}{2}[\log(2\pi e(N-\delta)\text{sin}^2 \theta)-\epsilon]} \leq |C| \leq 2^{\frac{m}{2}[\log(2\pi e(N+\delta)\text{sin}^2 \theta)+\epsilon]}\label{AE:shellcapvol}
\end{align}
where $\epsilon \to 0$ as $m \to \infty$.

\subsection{Surface Area (Volume) of the Intersection of Two Spherical (Shell) Caps}\label{AS:twocap}

Recall $\mathbb{S}^{m-1} \subset \mathbb{R}^m$ is the $(m-1)$-sphere of radius $R=\sqrt{mN}$. Let
$$C_i = \text{Cap}({\mathbf v_i},\theta_i) = \{{\bf v}\in \mathbb{S}^{m-1}: \angle({\bf v},{\mathbf v_i}) \leq \theta_i\}, \;  i=1,2$$
be two spherical caps on $\mathbb{S}^{m-1}$ such that $\angle({\mathbf v_1},{\mathbf v_2}) = \frac{\pi}{2}$, $\theta_i\leq\frac{\pi}{2}$, and $\theta_1+\theta_2 > \frac{\pi}{2}$. We have the following lemma that characterizes the intersection measure $\mu(C_1 \cap C_2)$ of these two caps.

\begin{lemma}\label{L:surfaceintersectioncap}
For any $\epsilon>0$ there exists an $M(\epsilon)$ such that for $m>M(\epsilon)$,
\begin{align*}
 \mu(C_1 \cap C_2) \leq 2^{\frac{m}{2}[\log(2\pi eN(\text{sin}^2 \theta_1 - \cos^2 \theta_2))+\epsilon]}
\end{align*}
and
\begin{align*}
  \mu(C_1 \cap C_2) \geq 2^{\frac{m}{2}[\log(2\pi eN(\text{sin}^2 \theta_1 - \cos^2 \theta_2))-\epsilon]} .
\end{align*}
\end{lemma}

\begin{proof}
To prove this lemma, we will first derive the surface area formula for the intersection of the above two caps (see also \cite{twocap}), and then characterize the exponent of this area.

\underline{Deriving the Surface Area Formula:} 
Consider the points ${\bf v}\in \mathbb{S}^{m-1}$ such that
$$\angle({\mathbf v_1},{\bf v}) = \theta_1$$
and
$$\angle({\mathbf v_2},{\bf v}) = \theta_2 .$$
These points satisfy the linear relations
$$\langle {\mathbf v_1},{\bf v} \rangle = R^2 \cos\theta_1$$
and
$$\langle {\mathbf v_2},{\bf v} \rangle = R^2 \cos\theta_2,$$
 and therefore all such ${\bf v}$ lie in the unique $m-1$ dimensional subspace $H$ defined by
$$\left\langle \frac{{\mathbf v_1}}{\cos\theta_1} - \frac{{\mathbf v_2}}{\cos\theta_2}, {\bf v} \right\rangle = 0  .$$
The angle between the hyperplane $H$ and the vector ${\mathbf v_2}$ is
{\small \begin{eqnarray*}
\phi = \frac{\pi}{2} - \arccos\left(\frac{1}{R\sqrt{\frac{1}{\cos^2\theta_1}+\frac{1}{\cos^2\theta_2}}}\left\langle \frac{{\mathbf v_1}}{\cos\theta_1} - \frac{{\mathbf v_2}}{\cos\theta_2}, {\mathbf v_2} \right\rangle\right)
\end{eqnarray*}}
and because ${\mathbf v_1}$ and ${\mathbf v_2}$ are orthogonal and $\|{\mathbf v_2}\| = R$,
\begin{align*}
\phi &= \frac{\pi}{2} - \arccos\left(\frac{1}{\cos\theta_2\sqrt{\frac{1}{\cos^2\theta_1}+\frac{1}{\cos^2\theta_2}}}\right) \\
& = \arctan \left(\frac{\cos\theta_1}{\cos\theta_2}\right) .
\end{align*}

The approach will be as follows. Divide the intersection $C_1 \cap C_2$ into two parts $C^+$ and $C^-$ that are on either side of the hyperplane $H$. More concretely,
$$C^+ = \left\{ {\bf v}\in C_1 \cap C_2 : \left\langle {\bf v},\frac{{\mathbf v_1}}{\cos\theta_1} - \frac{{\mathbf v_2}}{\cos\theta_2}\right\rangle \geq 0 \right\}$$
and
$$C^- = \left\{ {\bf v}\in C_1 \cap C_2 : \left\langle {\bf v},\frac{{\mathbf v_1}}{\cos\theta_1} - \frac{{\mathbf v_2}}{\cos\theta_2}\right\rangle < 0 \right\} .$$
Each part $C^+$ and $C^-$ can be written as a union of lower dimensional spherical caps. We will find the measure of each part by integrating the measures of these lower dimensional caps.

The measure of the cap $C_2$ can be expressed as the integral
$$\mu(C_2) = \int_0^{\theta_2} A_{m-2}(R\sin\rho)Rd\rho$$
where $A_{m-2}(R\sin\rho)$ is the surface area of the $(m-2)$-sphere with radius $R\sin\rho$. If we consider a single $(m-2)$-sphere at some angle $\rho$, then the hyperplane $H$ divides that $(m-2)$-sphere into two spherical caps. The claim is that each of these $m-2$ dimensional caps that is on the side of $H$ with ${\mathbf v_1}$ is contained in $C^+$ (and those on the side with ${\mathbf v_2}$ are contained in $C^-$). Furthermore, all points in $C^+$ are in one of these $m-2$ dimensional caps.
The claim follows because
$$\left\langle {\bf v},\frac{{\mathbf v_1}}{\cos\theta_1} - \frac{{\mathbf v_2}}{\cos\theta_2}\right\rangle \geq 0$$ implies
$$\cos\theta_2\cos\big(\angle({\bf v},{\mathbf v_1})\big) \geq \cos\theta_1\cos\big(\angle({\bf v},{\mathbf v_2})\big)$$
and since $\angle({\bf v},{\mathbf v_2}) \leq \theta_2$ and $\cos\big(\angle({\bf v},{\mathbf v_2})\big) \geq \cos\theta_2$, this implies
$$\cos\theta_2\cos\big(\angle({\bf v},{\mathbf v_1})\big) \geq \cos\theta_1\cos\theta_2  .$$
Finally, this implies $\angle({\bf v},{\mathbf v_1})\leq \theta_1$, ${\bf v} \in C_1$, and ${\bf v} \in C^+$.

Note that for $\rho<\phi$, the $(m-2)$-sphere at angle $\rho$ is entirely on the ${\mathbf v_2}$ side of $H$, and does not need to be included when computing the measure of $C^+$. This establishes the fact that
$$\mu(C^+) = \int_{\phi}^{\theta_2} C_{m-2}^{\theta_\rho}(R\sin\rho)Rd\rho$$
where $C_{m-2}^{\theta_\rho}(R\sin\rho)$ is the surface area of an $m-2$ dimensional spherical cap defined by angle $\theta_\rho$ on the $(m-2)$-sphere of radius $R\sin\rho$. Writing
$$\cos \theta_\rho = \frac{h}{R\sin\rho}$$
note that $h$ is the distance from the center of the $(m-2)$-sphere at angle $\rho$ to the $m-2$ dimensional hyperplane that divides the sphere into two caps. Furthermore, since the $(m-2)$-sphere has center $(R\cos\rho){\mathbf v_2}$, we have
$$\tan \phi = \frac{h}{R\cos\rho} \; .$$
Therefore,
$$\theta_\rho = \arccos\left(\frac{\text{tan}\phi}{\text{tan}\rho}\right) \; .$$
Combining this with the corresponding result for $\mu(C^-)$ yields
\begin{align*}
\mu(C_1\cap C_2) & =   \mu(C^+)+\mu(C^-) \\
& =   \int_{ \phi}^{\theta_2} C_{m-2}^{\arccos\left(\frac{\text{tan} \phi}{\text{tan}\rho}\right)}(R\sin\rho)R d\rho \\
& \;~~~~  + \int_{\frac{\pi}{2} -  \phi}^{\theta_1} C_{m-2}^{\arccos\left(\frac{\text{tan} (\pi/2-\phi)}{\text{tan}\rho}\right)}(R\sin\rho)R d\rho.
\end{align*}
This expression can be rewritten using known expressions for the area of a spherical cap in terms of the regularized incomplete beta function as
\begin{eqnarray*}
\mu(C_1 \cap C_2)   =J(\phi, \theta_2)+J(\pi/2-\phi, \theta_1),
\end{eqnarray*}
where $J(\phi, \theta_2)$ is defined as
{ \begin{align}
&J(\phi, \theta_2)\nonumber \\
&=\frac{{(\pi mN)}^{\frac{m-1}{2}}}{\Gamma\left(\frac{m-1}{2}\right)} \int_{\phi}^{\theta_2} (\sin^{m-2}\rho ) I_{1-\left(\frac{\text{tan}\phi}{\text{tan} \rho}\right)^2}\left(\frac{m-2}{2},\frac{1}{2} \right)d\rho \label{E:defJ}
\end{align}}
and $J(\pi/2-\phi, \theta_1)$ is defined similarly. Here in \dref{E:defJ}, $I_{x}(a,b)$ is the regularized incomplete beta function, given by
\begin{align}
I_{x}(a,b)=\frac{B(x;a,b)}{B(a,b)},\label{E:reg}
\end{align}
where $B(x;a,b)$ and $B(a,b)$ are the incomplete beta function and the complete beta function respectively:
\begin{align*}
B(x;a,b)&=\int_{0}^{x}t^{a-1}(1-t)^{b-1}dt\\
B(a,b)&=\frac{\Gamma(a)\Gamma(b)}{\Gamma(a+b)}.
\end{align*}

\underline{Characterizing the Exponent:} We now lower and upper bound $J(\phi, \theta_2)$ with exponential functions. First, using Stirling's approximation, $\frac{{(\pi mN)}^{\frac{m-1}{2}}}{\Gamma\left(\frac{m-1}{2}\right)}$ on the R.H.S. of \dref{E:defJ} can be bounded as
\begin{align} \label{E:sphere}
2^{\frac{m}{2}[\log (2\pi e N)-\epsilon_1]}\leq \frac{{(\pi mN)}^{\frac{m-1}{2}}}{\Gamma\left(\frac{m-1}{2}\right)}\leq 2^{\frac{m}{2}[\log (2\pi e N)+\epsilon_1]}
\end{align}
for some $\epsilon_1 \to 0$ as $m\to \infty$.

Now consider $$I_{1-\left(\frac{\text{tan}\phi}{\text{tan} \rho}\right)^2}\left(\frac{m-2}{2},\frac{1}{2} \right)$$ inside the integral on the R.H.S. of \dref{E:defJ}. In light of \dref{E:reg}, it can be written as
\begin{align}
I_{1-\left(\frac{\text{tan}\phi}{\text{tan} \rho}\right)^2}\left(\frac{m-2}{2},\frac{1}{2} \right)=\frac{B\left(1-\left(\frac{\text{tan}\phi}{\text{tan} \rho}\right)^2;\frac{m-2}{2},\frac{1}{2}\right)}{B\left(\frac{m-2}{2},\frac{1}{2}\right)}.\label{E:numden}
\end{align}
For the denominator in \dref{E:numden}, by Stirling's approximation, we have
\begin{align*}
B\left(\frac{m-2}{2},\frac{1}{2}\right)\sim \Gamma\left(\frac{1}{2}\right)\left(\frac{m-2}{2}\right)^{-\frac{1}{2}}.
\end{align*}
For the numerator in \dref{E:numden}, we have
\begin{align*}
&B\left(1-\left(\frac{\text{tan}\phi}{\text{tan} \rho}\right)^2;\frac{m-2}{2},\frac{1}{2}\right)\\
&=\int_{0}^{1-\left(\frac{\text{tan}\phi}{\text{tan} \rho}\right)^2} t^{\frac{m-4}{2}}(1-t)^{-\frac{1}{2}}dt\\
&\geq\int_{0}^{1-\left(\frac{\text{tan}\phi}{\text{tan} \rho}\right)^2} t^{\frac{m-4}{2}} dt\\
&=\frac{2}{m-2}  t^{\frac{m-2}{2}} \big|_{0}^{1-\left(\frac{\text{tan}\phi}{\text{tan} \rho}\right)^2} \\
&=\frac{2}{m-2}  \left[{1-\left(\frac{\text{tan}\phi}{\text{tan} \rho}\right)^2}\right]^{\frac{m-2}{2}}  \\
&\geq 2^{\frac{m}{2}\left[\log\left(1-\left(\frac{\text{tan}\phi}{\text{tan} \rho}\right)^2 \right)  -\epsilon_2\right]},
\end{align*}
for some $\epsilon_2 \to 0$ as $m \to \infty$, and
\begin{align*}
&B\left(1-\left(\frac{\text{tan}\phi}{\text{tan} \rho}\right)^2;\frac{m-2}{2},\frac{1}{2}\right)\\
&=\int_{0}^{1-\left(\frac{\text{tan}\phi}{\text{tan} \rho}\right)^2} t^{\frac{m-4}{2}}(1-t)^{-\frac{1}{2}}dt\\
&\leq\int_{0}^{1-\left(\frac{\text{tan}\phi}{\text{tan} \rho}\right)^2} t^{\frac{m-4}{2}} \left(1- \left(1-\left(\frac{\text{tan}\phi}{\text{tan} \rho}\right)^2 \right)  \right)^{-\frac{1}{2}} dt\\
&=\frac{\text{tan} \rho} {\text{tan}\phi}\int_{0}^{1-\left(\frac{\text{tan}\phi}{\text{tan} \rho}\right)^2} t^{\frac{m-4}{2}}  dt\\
&\leq\frac{\text{tan} \theta_2} {\text{tan}\phi}\int_{0}^{1-\left(\frac{\text{tan}\phi}{\text{tan} \rho}\right)^2} t^{\frac{m-4}{2}}  dt\\
&=\frac{2\text{tan} \theta_2}{(m-2)\text{tan}\phi}  \left[{1-\left(\frac{\text{tan}\phi}{\text{tan} \rho}\right)^2}\right]^{\frac{m-2}{2}}  \\
&\leq 2^{\frac{m}{2}\left[\log\left(1-\left(\frac{\text{tan}\phi}{\text{tan} \rho}\right)^2 \right)  +\epsilon_3\right]},
\end{align*}
for some $\epsilon_3 \to 0$ as $m \to \infty$. Also noting that
\begin{align*}
\sin^{m-2}\rho = 2^{\frac{m-2}{2}\log\text{sin}^2\rho}
\end{align*}
with $\rho\in [\phi,\theta_2]$, we can bound the integrand in \dref{E:defJ} as
\begin{align*}
&\left(\sin^{m-2}\rho\right) I_{1-\left(\frac{\text{tan}\phi}{\text{tan} \rho}\right)^2}\left(\frac{m-2}{2},\frac{1}{2} \right)\\
&\geq 2^{\frac{m}{2}\left[\log\left((\text{sin}^2\rho)\left(1-\left(\frac{\text{tan}\phi}{\text{tan} \rho}\right)^2 \right)\right)  -\epsilon_4\right]}\\
&= 2^{\frac{m}{2}\left[\log\left( \text{sin}^2\rho-\text{tan}^2\phi\cos^2\rho   \right)  -\epsilon_4\right]}
\end{align*}
and
\begin{align*}
&\left(\sin^{m-2}\rho\right) I_{1-\left(\frac{\text{tan}\phi}{\text{tan} \rho}\right)^2}\left(\frac{m-2}{2},\frac{1}{2} \right)\\
&\leq  2^{\frac{m}{2}\left[\log\left( \text{sin}^2\rho-\text{tan}^2\phi\cos^2\rho   \right)  +\epsilon_4\right]}
\end{align*}
for some $\epsilon_4 \to 0$ as $m \to \infty$.  For sufficiently large $m$,
\begin{align*}
&\int_{\phi}^{\theta_2} (\sin^{m-2}\rho ) I_{1-\left(\frac{\text{tan}\phi}{\text{tan} \rho}\right)^2}\left(\frac{m-2}{2},\frac{1}{2} \right)d\rho \\
&\geq  \int_{\theta_2-\frac{1}{m}}^{\theta_2} (\sin^{m-2}\rho ) I_{1-\left(\frac{\text{tan}\phi}{\text{tan} \rho}\right)^2}\left(\frac{m-2}{2},\frac{1}{2} \right)d\rho \\
&\geq \int_{\theta_2-\frac{1}{m}}^{\theta_2}   2^{\frac{m}{2}\left[\log\left( \text{sin}^2\rho-\text{tan}^2\phi\cos^2\rho   \right)   -\epsilon_4\right]}d\rho\\
&\geq \frac{1}{m}  2^{\frac{n}{2}\left[\log\left( \text{sin}^2(\theta_2-\frac{1}{m})-\text{tan}^2\phi\cos^2(\theta_2-\frac{1}{m})   \right)   -\epsilon_4\right]} \\
&\geq  2^{\frac{m}{2}\left[\log\left( \text{sin}^2\theta_2-\text{tan}^2\phi\cos^2\theta_2   \right)   -\epsilon_5\right]} \\
&= 2^{\frac{m}{2}\left[\log\left( \text{sin}^2\theta_2- \cos^2\theta_1   \right)   -\epsilon_5\right]},
\end{align*}
and
\begin{align*}
&\int_{\phi}^{\theta_2} (\sin^{m-2}\rho ) I_{1-\left(\frac{\text{tan}\phi}{\text{tan} \rho}\right)^2}\left(\frac{m-2}{2},\frac{1}{2} \right)d\rho \\
&\leq 2^{\frac{m}{2}\left[\log\left( \text{sin}^2\theta_2- \cos^2\theta_1   \right)   +\epsilon_5\right]}
\end{align*}
for some $\epsilon_5 \to 0$ as $m\to \infty$.

Combining this with \dref{E:sphere}, we can bound $J(\phi, \theta_2)$ as
\begin{align*}
 2^{\frac{m}{2}\left[\log2\pi e N\left( \text{sin}^2\theta_2- \cos^2\theta_1   \right)   -\epsilon_6\right]} & \leq J(\phi, \theta_2) \\
& \leq 2^{\frac{m}{2}\left[\log2\pi e N \left( \text{sin}^2\theta_2- \cos^2\theta_1   \right)   +\epsilon_6\right]}
\end{align*}
for some $\epsilon_6 \to 0$ as $m\to \infty$.

Due to symmetry, we can also bound $J(\pi/2-\phi, \theta_1)$ as
\begin{align*}
 2^{\frac{m}{2}\left[\log2\pi e N\left( \text{sin}^2\theta_1- \cos^2\theta_2   \right)   -\epsilon_6\right]} & \leq J(\pi/2-\phi, \theta_1) \\
& \leq 2^{\frac{m}{2}\left[\log2\pi e N\left( \text{sin}^2\theta_1- \cos^2\theta_2   \right)   +\epsilon_6\right]}.
\end{align*}
Noting that $\text{sin}^2\theta_2- \cos^2\theta_1 =\text{sin}^2\theta_1- \cos^2\theta_2$, we have
\begin{align*}
\mu(C_1 \cap C_2)&\geq J(\phi, \theta_2)+J(\pi/2-\phi, \theta_1)\\
&\geq  2^{\frac{m}{2}\left[\log2\pi e N\left( \text{sin}^2\theta_1- \cos^2\theta_2   \right)   -\epsilon\right]}
\end{align*}
and
\begin{align*}
\mu(C_1 \cap C_2)\leq   2^{\frac{m}{2}\left[\log2\pi e N\left( \text{sin}^2\theta_1- \cos^2\theta_2   \right)   +\epsilon\right]}
\end{align*}
for some $\epsilon  \to 0$ as $m\to \infty$. This completes the proof of the lemma.
\end{proof}

 \bigbreak
 
We now utilize Lemma \ref{L:surfaceintersectioncap} to characterize the volume of the intersection of two shell caps. Consider a spherical shell
$$\mbox{Shell}\left(\mathbf 0, R_L, R_U  \right)$$
with $R_L = \sqrt{m(N-\delta)}$, $R_U=\sqrt{m(N+\delta)}$ and two caps on this shell, i.e. $S_1 = \text{ShellCap}(\mathbf z_0, \theta)$ and $S_2 = \text{ShellCap}(\mathbf y_0, \omega)$, where $\angle(\mathbf z_0,\mathbf y_0)=\pi/2$ and $\theta+\omega>\pi/2$. The following lemma bounds the intersection volume $|S_1\cap S_2|$ of these two shell caps.

\begin{lemma}\label{L:volumeintersectionshellcap}
For any $\epsilon>0$ there exists an $M(\epsilon)$ such that for $m>M(\epsilon)$,
\begin{align*}
 |S_1\cap S_2| \geq  2^{\frac{m}{2}[\log(2\pi eN(\text{sin} ^2 \theta - \cos^2 \omega))-\epsilon]}
\end{align*}
and
\begin{align*}
 |S_1\cap S_2|  \leq 2^{\frac{m}{2}[\log(2\pi e(N+\delta)(\text{sin}^2 \theta  - \cos^2 \omega))+\epsilon]} .
\end{align*}
\end{lemma}

\begin{proof}
Using spherical coordinates, we have
\begin{align} \label{eq:lemmaC2_2}
& |S_1\cap S_2| = \int_{\mathbb R^m} 1_{S_1\cap S_2}(\mathbf z) \; d\mathbf z \nonumber\\
& =   \int_{\mathbb S^{m-1}} \left(\int_{R_L}^{R_U} \left(\frac{r}{R}\right)^{m-1}1_{S_1\cap S_2}\left(\frac{r}{R}\mathbf z\right) \; dr \right) d\mathbf z \nonumber\\
& =   \int_{\mathbb S^{m-1}} 1_{\text{Cap}(\mathbf z_0,\theta) \cap \text{Cap}(\mathbf y_0, \omega)}(\mathbf z)d\mathbf z \int_{R_L}^{R_U} \left(\frac{r}{R}\right)^{m-1}dr \nonumber\\
& =   \mu(\text{Cap}(\mathbf z_0,\theta) \cap \text{Cap}(\mathbf y_0, \omega)) \int_{R_L}^{R_U} \left(\frac{r}{R}\right)^{m-1}dr \;
\end{align}
where the integral term on the right is bounded as
\begin{align} \label{eq:lemmaC2_1}
\int_{\sqrt{m(N-\delta)}}^{\sqrt{m(N+\delta)}} \left(\frac{r}{R}\right)^{m-1}dr & \geq \int_{\sqrt{mN}}^{\sqrt{m(N+\delta)}} \left(\frac{r}{R}\right)^{m-1}dr \nonumber\\
& \geq  \sqrt{m(N+\delta)}-\sqrt{mN} \; .
\end{align}
Given $\epsilon >0$, set $M = \max \{M_1,M_2\}$ where $M_1$ is given by Lemma \ref{L:surfaceintersectioncap} to ensure
$$ \mu(\text{Cap}(\mathbf z_0,\theta) \cap \text{Cap}(\mathbf y_0, \omega)) \geq 2^{\frac{m}{2}[\log(2\pi eN(\text{sin}^2 \theta - \cos^2 \omega))-\epsilon/2]}$$
and $M_2$ is chosen to be sufficiently large so that the right-hand side of \eqref{eq:lemmaC2_1} satisfies
$$\sqrt{m(N+\delta)}-\sqrt{mN} \geq 2^{-m\epsilon} \; .$$
Together with \eqref{eq:lemmaC2_2}, this implies
$$|S_1\cap S_2| \geq 2^{\frac{m}{2}[\log(2\pi eN(\text{sin}^2 \theta - \cos^2 \omega))-\epsilon]}$$ for $m>M$.

For the inequality in the other direction, define $\mathbb{S}_{R_U}^{m-1}$ to be the $m-1$ sphere of radius $R_U$ with Haar measure $\mu_{R_U}$. Then
\begin{align} \label{eq:lemmaC2_4}
& |S_1\cap S_2| = \int_{\mathbb R^m} 1_{S_1\cap S_2}(\mathbf z) \; d\mathbf z \nonumber\\
& =   \int_{\mathbb S^{m-1}_{R_U}} \left(\int_{R_L}^{R_U} \left(\frac{r}{R_U}\right)^{m-1}1_{S_1\cap S_2}\left(\frac{r}{R_U}\mathbf z\right) \; dr \right) d\mathbf z \nonumber\\
& =   \int_{\mathbb S^{m-1}_{R_U}} 1_{\text{Cap}(\mathbf z_0,\theta) \cap \text{Cap}(\mathbf y_0, \omega)}(\mathbf z)d\mathbf z \int_{R_L}^{R_U} \left(\frac{r}{R_U}\right)^{m-1}dr \nonumber\\
& =   \mu_{R_U}(\text{Cap}(\mathbf z_0,\theta) \cap \text{Cap}(\mathbf y_0, \omega))\int_{R_L}^{R_U} \left(\frac{r}{R_U}\right)^{m-1}dr \; 
\end{align}
where the integral term on the right is bounded as
\begin{align} \label{eq:lemmaC2_5}
\int_{\sqrt{m(N-\delta)}}^{\sqrt{m(N+\delta)}} \left(\frac{r}{R_U}\right)^{m-1}dr &\leq  \sqrt{m(N+\delta)}-\sqrt{m(N-\delta)} \; .
\end{align}
Given $\epsilon >0$, set $M = \max \{M_1,M_2\}$ where $M_1$ is given by Lemma \ref{L:surfaceintersectioncap} to ensure
\begin{align*}
\mu_{R_U} & (\text{Cap}(\mathbf z_0,\theta) \cap \text{Cap}(\mathbf y_0, \omega)) \\
& \leq 2^{\frac{m}{2}[\log(2\pi e(N+\delta)(\text{sin}^2 \theta_1 - \cos^2 \theta_2))+\epsilon/2]}
\end{align*}
and $M_2$ is chosen to be sufficiently large so that the right-hand side of \eqref{eq:lemmaC2_5} satisfies
$$\sqrt{m(N+\delta)}-\sqrt{m(N-\delta)} \leq 2^{m\epsilon} \; .$$
Together with \eqref{eq:lemmaC2_4}, this implies
$$|S_1 \cap S_2| \leq 2^{\frac{m}{2}[\log(2\pi e(N+\delta)(\text{sin}^2 \theta - \cos^2 \omega))+\epsilon]}$$ for $m>M$.
\end{proof}

\subsection{Volume of the Intersection of Two Balls }\label{A:twoball}
\begin{proof}[Proof of Lemma \ref{L:twoball} ]
The intersection of $\text{Ball}(\mathbf c_1, \sqrt{mR_1})$ and $\text{Ball}(\mathbf c_2, \sqrt{mR_1})$ consists of two caps: $C_1$ and $C_2$, as depicted in Fig. \ref{F:newtwoball}. To bound the volume of $\text{Ball}(\mathbf c_1, \sqrt{mR_1}) \cap \text{Ball}(\mathbf c_2, \sqrt{mR_1})$, we will bound $|C_1|$ and $|C_2|$ respectively.

\begin{figure}
\centering
\includegraphics[width=0.4\textwidth]{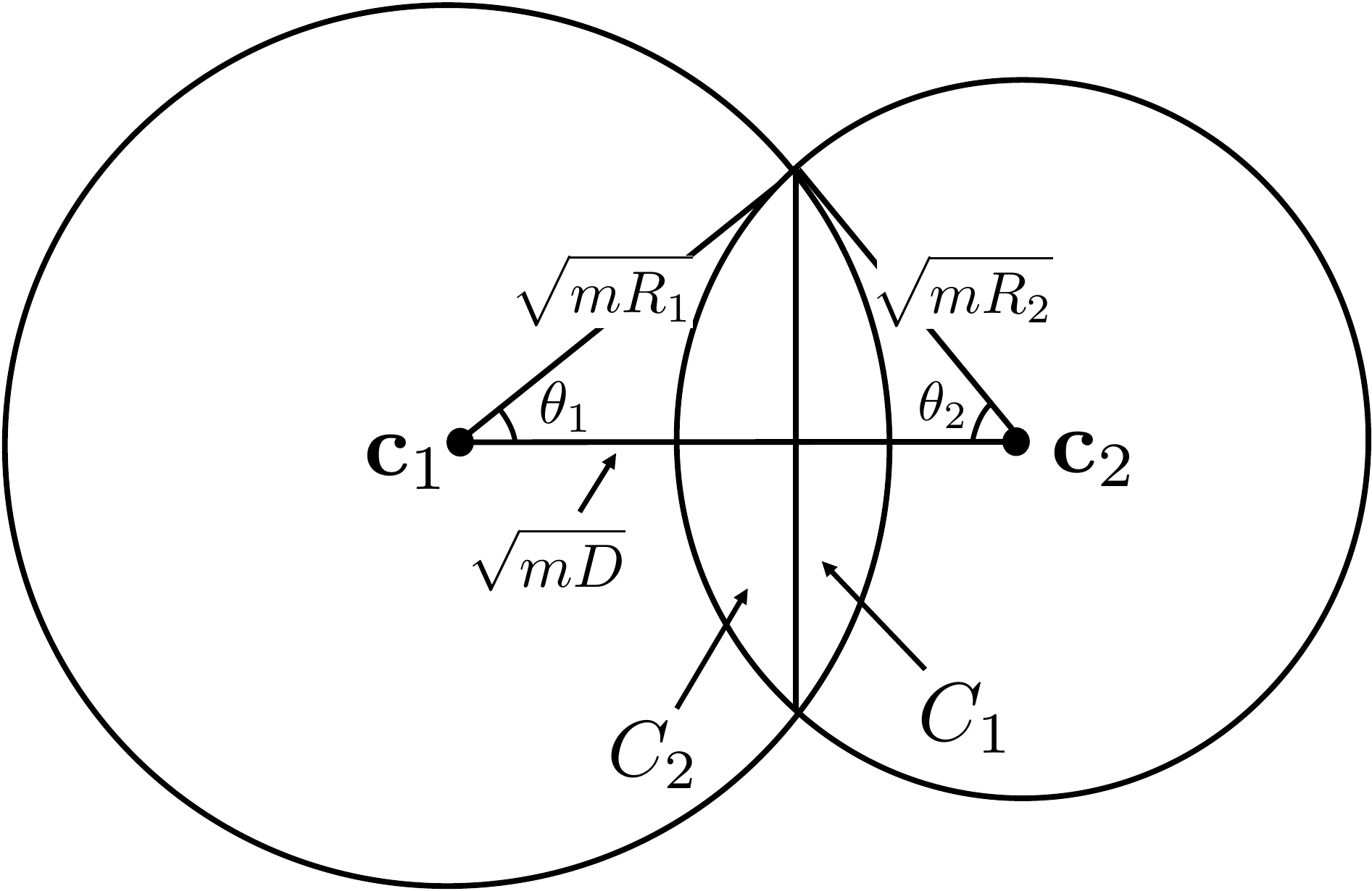}
\caption{Intersection of two balls.}
\label{F:newtwoball}
\end{figure}

We first bound $|C_1|$. By the cosine formula, we have
\begin{align*}
\cos \theta_1& = \frac{mR_1+mD-mR_2 }{2\sqrt{mR_1}\sqrt{mD}}\\
&=\frac{ R_1+ D- R_2 }{2\sqrt{ R_1D} }
\end{align*}
and therefore
\begin{align*}
\text{sin}^2 \theta_1 &=1-\cos^2 \theta_1\\
 &= 1- \frac{(R_1+D-R_2)^2 }{4 R_1D}\\
 &= \frac{2R_1D+2R_1R_2+2DR_2-R_1^2-R_2^2-D^2}{4R_1D}.
\end{align*}
From Appendix \ref{AS:singlecap}, we have for any $\epsilon>0$ and $m$ sufficiently large,
\begin{align*}
|C_1|&\leq 2^{m\left(\frac{1}{2}\log 2\pi e R_1 \text{sin}^2 \theta_1+\frac{\epsilon}{2}   \right)}\\
&=2^{m\left(\frac{1}{2}\log \pi e \lambda(R_1,R_2,D)    +\frac{\epsilon}{2} \right)}
\end{align*}
where
{\small \begin{align*}
\lambda(R_1,R_2,D):=\frac{ 2R_1D+2R_1R_2+2DR_2-R_1^2-R_2^2-D^2 }{2D} .
\end{align*}}

Similarly, we have
\begin{align*}
\text{sin}^2 \theta_2 &=1-\cos^2 \theta_2\\
 &= 1- \frac{(R_2+D-R_1)^2 }{4 R_2D}\\
 &= \frac{2R_1D+2R_1R_2+2DR_2-R_1^2-R_2^2-D^2}{4R_2D}
\end{align*}
and therefore
\begin{align*}
|C_2|&\leq 2^{m\left(\frac{1}{2}\log 2\pi e R_2 \text{sin}^2 \theta_2+\frac{\epsilon}{2}   \right)}\\
&=2^{m\left(\frac{1}{2}\log \pi e \lambda(R_1,R_2,D)    +\frac{\epsilon}{2} \right)}.
\end{align*}

Combining the above, we obtain
\begin{align*}
& \left|\text{Ball} (\mathbf c_1, \sqrt{mR_1}) \cap \text{Ball}(\mathbf c_2, \sqrt{mR_1}) \right| \\
& = |C_1|+|C_2| \\
& \leq 2^{m\left(\frac{1}{2}\log \pi e \lambda(R_1,R_2,D)    +\epsilon \right)}
\end{align*}
for any $\epsilon>0$ and $m$ sufficiently large.
\end{proof}

\section*{Acknowledgement}
The authors would like to acknowledge inspiring discussions with Liang-Liang Xie within a preceding collaboration \cite{WuOzgurXie_TIT}. They would also like to thank the anonymous
reviewers and the Associate Editor for many valuable comments that helped improve the presentation of this paper.

\begin{IEEEbiographynophoto}{Xiugang Wu}
(M'14)  received the B.Eng. degree with honors in electronics and
information engineering from Tongji University, Shanghai, China, in 2007, and the M.A.Sc and Ph.D. degree in
electrical and computer engineering from the University of Waterloo, Waterloo,
Ontario, Canada, in 2009 and 2014, respectively. He was a postdoctoral fellow in the Department of Electrical Engineering, Stanford University, Stanford, CA, during 2015--2018. He is currently an assistant professor at the University of Delaware, Newark, DE, where he is jointly appointed in the Department of Electrical and Computer Engineering and the Department of Computer and Information Sciences. His research interests are in information theory, networks, data science, and the interplay between them. He is a recipient of the 2017 NSF Center for Science of Information (CSoI) Postdoctoral Fellowship.
\end{IEEEbiographynophoto}

\begin{IEEEbiographynophoto}{Leighton Pate Barnes}
(S'17) received a B.S. in Mathematics '13, B.S. in Electrical Science and Engineering '13, and M.Eng. in Electrical Engineering and Computer Science '15, all from the Massachusetts Institute of Technology. While there, he received the Harold L. Hazen Award for excellence in teaching. He is currently a Ph.D. candidate in the Department of Electrical Engineering at Stanford University, where he studies geometric extremal problems applied to information theory, communication, and estimation.
\end{IEEEbiographynophoto}

\begin{IEEEbiographynophoto}{Ayfer  \"{O}zg\"{u}r}
(M'06) received her B.Sc. degrees in electrical engineering and physics from Middle East Technical University, Turkey, in 2001 and the M.Sc. degree in communications from the same university in 2004. From 2001 to 2004, she worked as hardware engineer for the Defense Industries Development Institute in Turkey. She received her Ph.D. degree in 2009 from the Information Processing Group at EPFL, Switzerland. In 2010 and 2011, she was a post-doctoral scholar at the same institution. She is currently an Assistant Professor in the Electrical Engineering Department at Stanford University where she is a Hoover and Gabilan Fellow. Her current research interests include distributed communication and learning, wireless systems, and information theory. Dr.  \"{O}zg\"{u}r received the EPFL Best Ph.D. Thesis Award in 2010, the NSF CAREER award in 2013 and the Okawa Foundation Research Grant in 2018.
\end{IEEEbiographynophoto}


\begin{thebibliography}{10}


\bibitem{WuBarnesOzgur}
X. Wu, L. Barnes, A. Ozgur, ``Cover's open problem: ``The capacity of the relay channel'',''  {\em  Proc. of  54th Annual Allerton Conference on Communication, Control, and Computing}, Allerton Retreat Center, Monticello, Illinois,  2016.


\bibitem{cg87}
T. M. Cover,  ``The capacity of the relay channel,''   \emph{Open Problems in
Communication and Computation}, edited by T. M. Cover and B. Gopinath, Eds.
New York: Springer-Verlag, 1987, pp. 72--73.



\bibitem{IZS2016}
X. Wu and A. Ozgur, ``Improving on the cut-set bound via geometric analysis of typical sets,'' in {\em
Proc. of 2016 International Zurich Seminar on Communications}.

\bibitem{WuOzgurXie_TIT}
X. Wu, A. Ozgur, L.-L. Xie,  ``Improving on the cut-set bound via geometric analysis of typical sets,''   {\em
{IEEE} Trans. Inform. Theory}, vol. 63, pp. 2254--2277, April 2017.


\bibitem{Allerton2015}
X. Wu and A. Ozgur, ``Cut-set bound is loose for Gaussian relay networks,'' in {\em  Proc. of  53rd Annual Allerton Conference on Communication, Control, and Computing}, Allerton Retreat Center, Monticello, Illinois, Sept. 29--Oct. 1, 2015.

\bibitem{WuOzgur_TIT_Gaussian}
X. Wu and A. Ozgur,   ``Cut-set bound is loose for Gaussian relay networks,'' {\em
{IEEE} Trans. Inform. Theory}, vol. 64, pp. 1023--1037, February 2018. 

\bibitem{ISIT2016}
X. Wu and A. Ozgur, ``Improving on the cut-set bound for general primitive relay channels,'' in \emph{Proc. of IEEE Int. Symposium on Information Theory}, Barcelona, Spain, Jul. 2016.

\bibitem{ISIT2017}
X. Wu, L. Barnes and A. Ozgur, ``The geometry of the relay channel,'' in \emph{Proc. of IEEE Int. Symposium on Information Theory}, Aachen, Germany, June 2017.

\bibitem{Allerton2017}
L. Barnes, X. Wu and A. Ozgur, ``A solution to Cover's problem for the binary symmetric relay channel: geometry of sets on the
Hamming sphere,'' in {\em  Proc. of  55th Annual Allerton Conference on Communication, Control, and Computing}, Allerton Retreat Center, Monticello, Illinois, Oct. 2017.

\bibitem{covelg79}
T.~Cover and A.~{El Gamal}, ``Capacity theorems for the relay channel,'' {\em
  {IEEE} Trans. Inform. Theory}, vol.~25, pp.~572--584, 1979.
  
\bibitem{Zhang}
Z. Zhang,  ``Partial converse for a relay channel,''  {\em
{IEEE} Trans. Inform. Theory}, vol.~34, no. 5,  pp.~1106--1110, Sept. 1988.

\bibitem{mono}
M. Raginsky and I. Sason, ``Concentration of measure inequalities in information theory, communications and coding,'' {\em Foundations and Trends in Communications and Information Theory}, vol.~10, no.~1--2, pp.~1--250,
second edition, October~2014.

\bibitem{shortcourse}
A. Burchard, \emph{A short course on rearrangement inequalities}, June 2009. Available: http://www.math.utoronto.ca/almut/rearrange.pdf

\bibitem{levy}
P. Levy,  \emph{Problmes concrets d'analyse fonctionnelle} (French), 2d ed, Gauthier-Villars, Paris, 1951.


\bibitem{matou}
J. Matou{\v{s}}ek,  \emph{Lectures on discrete geometry}, Volume 212, Springer Science \& Business Media, 2002.



\bibitem{isoperimetric}
G. Schechtman, ``Concentration, results and applications,'' \emph{Handbook of
the geometry of Banach spaces}, Vol. 2, 1603--1634, North-Holland, Amsterdam,
2003.




\bibitem{shannon49}
C. E. Shannon, ``Communication in the presence of noise,'' \emph{Proc. IRE},
vol. 37, pp. 10--21, Jan. 1949.

\bibitem{Shannon1948}
C. E. Shannon, ``A mathematical theory of communication,''
\emph{Bell Syst.\ Tech.\ J.}, vol.\ 27, pt.~I, pp.~379--423, 1948;
     pt.~II, pp.~623--656, 1948.


\bibitem{CoverBook}
T. Cover and J. Thomas, \emph{Elements of Information Theory}, 2nd ed. New York, NY, USA: Wiley, 2006.


\bibitem{Marton}
K. Marton, ``Bounding $\bar{d}$-distance by informational divergence: a method to
prove measure concentration,'' {\em Annals of Probability}, vol.~24, no.~2, pp.~857--866, 1996.



\bibitem{ElGamalKim}
A. El Gamal and Y.-H. Kim, \emph{Network Information Theory}, Cambridge, U.K.: Cambridge University Press, 2012.


\bibitem{Talagrand}
M. Talagrand, ``Transportation cost for Gaussian and other product measures,'' {\em Geometric \& Functional Analysis}, pp. 587--600.

\bibitem{singlecap}
S. Li, ``Concise formulas for the area and volume of a hyperspherical cap,'' \emph{Asian Journal of Mathematics and Statistics}, vol. 4, pp. 66--70, 2011.

\bibitem{twocap}
Y. Lee and W. C. Kim, ``Concise formulas for the surface area of the intersection of two hyperspherical caps,'' \emph{KAIST Technical Report}, 2014.


\bibitem{baernstein}
A. Baernstein II and B. A. Taylor, ``Spherical rearrangements, subharmonic functions, and $\ast$-functions in $n$-space," \emph{Duke Mathematical Journal}, vol. 43, no. 2, pp. 245-268, 1976.

\bibitem{Csiszarbook}
I. Csisz\'{a}r and J. K\"{o}rner, \emph{Information Theory: Coding Theorems for Discrete Memoryless Systems},
2nd ed. Cambridge University Press, 2011.




%
%
%
%
%
%
%
%
%
%
%
%
%
%
%
%
%
%
%
%

%
%
%






\end{thebibliography}
\end{document}